\documentclass[10pt,conference,letterpaper]{IEEEtran}
\IEEEoverridecommandlockouts

\usepackage{amsmath}
\usepackage{array}
\usepackage{cite}
\usepackage{amsmath,amssymb,amsfonts}
\usepackage{algorithmic}
\usepackage{graphicx}
\usepackage{amsthm} 
\usepackage{graphicx}
\usepackage{algorithm}
\usepackage{algorithmic}
\newtheorem{theorem}{Theorem}

\newtheorem{proposition}[theorem]{Proposition}

\newtheorem{remark}[theorem]{Remark}

\usepackage{graphicx} 
\usepackage{subcaption}
\usepackage{xcolor}
\usepackage{svg}
\usepackage{cite}
\usepackage{enumitem}
\setlist[enumerate]{topsep=0pt,itemsep=-1ex,partopsep=1ex,parsep=1ex,leftmargin=4ex}
\setlist[itemize]{topsep=0pt,itemsep=-1ex,partopsep=1ex,parsep=1ex,leftmargin=4ex}
\usepackage{caption}

\begin{document}
\title{Optimizing Age of Information without Knowing the Age of Information}
% \title{Minimizing Age of Information in Wireless Network without AoI Feedback}
% \title{Minimizing Age of Information without AoI Feedback: An Estimation Approach}
% \title{Minimizing Age of Information in Wireless Network without AoI Interaction}

\author{Zhuoyi Zhao and Igor Kadota% <-this % stops a space
\IEEEcompsocitemizethanks{\IEEEcompsocthanksitem Zhuoyi Zhao and Igor Kadota are with the Department of Electrical and Computer Engineering, Northwestern University, USA. E-mail: zhuoyizhao2025@u.northwestern.edu and kadota@northwestern.edu. 
}% <-this % stops an unwanted space} 
}
% \and
% \IEEEauthorblockN{Homer Simpson}
% \IEEEauthorblockA{Twentieth Century Fox\\
% Email: homer@thesimpsons.com}
% \and
% \IEEEauthorblockN{James Kirk\\ and Montgomery Scott}
% \IEEEauthorblockA{Starfleet Academy\\
% San Francisco, California 96678--2391\\
% Telephone: (800) 555--1212\\


%Telephone: (800) 555--1212, Fax: (888) 555--1212}


\maketitle

\begin{abstract}
Consider a network where a wireless base station (BS) connects multiple source-destination pairs. Packets from each source are generated according to a renewal process and are enqueued in a single-packet queue that stores only the freshest packet. The BS decides, at each time slot, which sources to schedule. Selected sources transmit their packet to the BS via unreliable links. Successfully received packets are forwarded to corresponding destinations. The connection between the BS and destinations is assumed unreliable and delayed. Information freshness is captured by the Age of Information (AoI) metric. The objective of the scheduling decisions is leveraging the delayed and unreliable AoI knowledge to keep the information fresh.

In this paper, we derive a lower bound on the achievable AoI by any scheduling policy. Then, we develop an optimal randomized policy for any packet generation processes. Next, we develop minimum mean square error estimators of the AoI and system times, and a Max-Weight Policy that leverages these estimators. We evaluate the AoI of the Optimal Randomized Policy and the Max-Weight Policy both analytically and through simulations. The numerical results suggest that the Max-Weight Policy with estimation outperforms the Optimal Randomized Policy even when the BS has no AoI knowledge.

\end{abstract}

\IEEEpeerreviewmaketitle

\section{Introduction}\label{sec:intro}

The Age of Information (AoI) metric has been receiving significant attention in the literature~\cite{kaul2011minimizing,dayal2023adaptive,AoI_V,AoIUAV1,9637803,yu2022age,beytur2020towards,abd2019role}
% ~\cite{AoIM/M/1,AoIsurvey,sun2017update,kaul2011minimizing,dayal2023adaptive,AoI_V,AoIUAV1,9637803,yu2022age,beytur2020towards,abd2019role,kadota2018scheduling,kadota2019scheduling,AoIDQN,AoIopenloop,fountoulakis2023scheduling,Whittle_Zhifeng,kadota2019minimizing,zakeri2023minimizing,saurav2023scheduling,talak2018optimizing,Vishrant_Multihop,bedewy2016optimizing,tang2020minimizing,arafa2017age,AoIbasicImplementation,WiFresh,AoIOFDMA,AoI5G,li2024aequitas,liu2023optimizing,liu2024optimizing,ji2024age,tripathi2019age,ayan2019age}
due to its relevance for emerging time-sensitive applications such as connected autonomous vehicles~\cite{kaul2011minimizing,dayal2023adaptive,AoI_V}, cooperative UAV swarms~\cite{AoIUAV1,9637803}, and the Internet-of-Things~\cite{yu2022age,beytur2020towards,abd2019role}.
%With the increasing demand for time-sensitive applications such as vehicular networks, UAV networks, and wireless sensor networks, the AoI has emerged as a critical metric of interest \cite{AoIM/M/1,AoIsurvey,Update or wait,kadota2018scheduling,AoIDQN,AoIopenloop,kadota2019minimizing,saurav2023scheduling,talak2018optimizing,AoIOFDMA,AoIbasicImplementation,WiFresh,AoI5G,li2024aequitas,liu2023optimizing,liu2024optimizing,ji2024age,tripathi2019age}. 
The AoI captures the freshness of information from the destination's perspective.
Specifically, the AoI measures the time elapsed since the generation of the latest information received by the destination. 
Consider a system in which information is timestamped upon generation.
%In systems where packets are timestamped upon packet generation, 
A higher timestamp indicates fresher information. 
Let $\tau^D(t)$ be the timestamp of the freshest packet received by the destination at time~$t$. 
The AoI is defined as $h(t) := t - \tau^D(t)$. 
While the destination does not receive a fresher packet, %When no new packets are received by the destination, 
$h(t)$ increases linearly, indicating that the information is becoming stale. 
When the destination receives a fresher packet, the timestamp $\tau^D(t)$ is updated, reducing the value of AoI $h(t)$. %, indicating that the information at the destination is fresher. 
%Notice that this AoI reduction is equal to the difference between the $\tau^D(t)$ and the timestamp of the received packet.

In this paper, we consider a network with multiple sources-destination pairs sharing time-sensitive information via a wireless base station (BS), as illustrated in Fig.~\ref{fig:Systemmodel}. 
Sources generate packets according to renewal processes. 
At every decision time~$t$, the transmission scheduling algorithm (running at the BS) selects a subset of sources for transmission over a shared and unreliable wireless channel. 
Packets that are successfully received by the BS are forwarded to the corresponding destinations. 
The connection between BS and destinations is unreliable and delayed. 
Our objective is to develop scheduling policies that attempt to keep the information at every destination as fresh as possible, i.e., that minimize AoI in the network. 

\begin{figure}[t]
    \centering
    \includegraphics[width=\columnwidth]{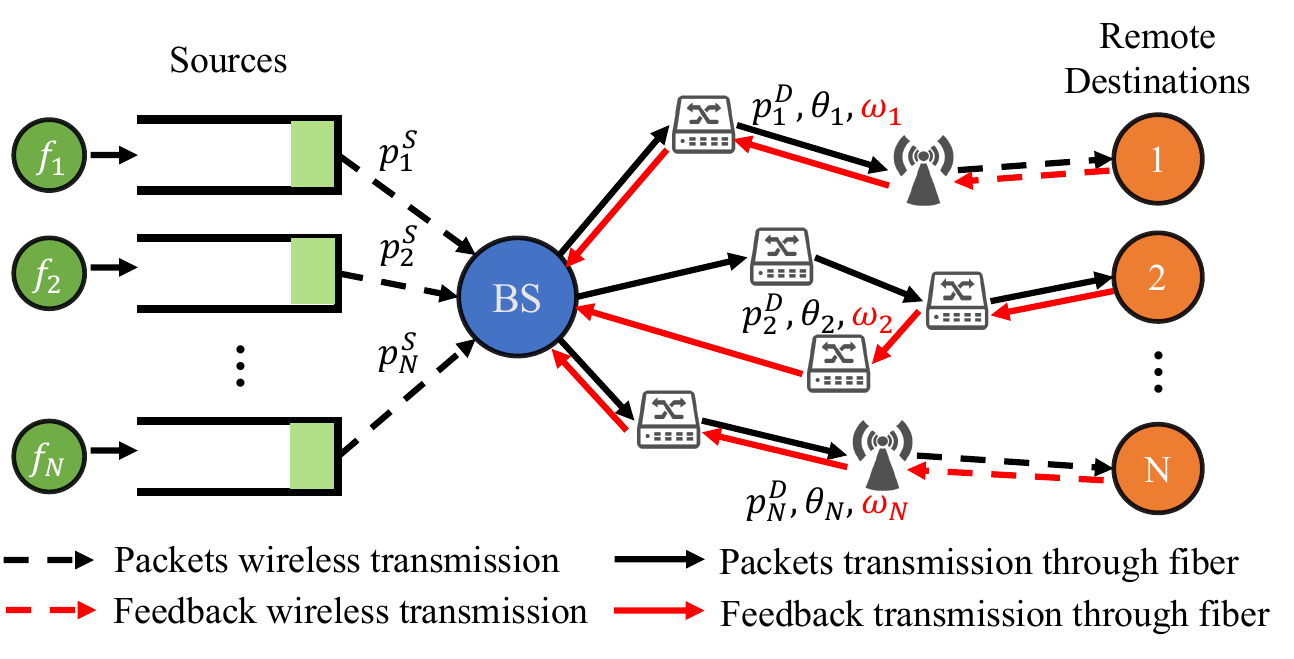}
\vspace{-1.5\baselineskip}
    \caption{Network with $N$ source-destination pairs sharing time-sensitive information via a wireless base station (BS). The BS selects $K$ sources at every decision time~$t$. Each selected source~$i$ transmits a packet to the BS through a wireless channel with reliability $p_i^S$. The BS forwards all successfully received packets to their corresponding destinations through a heterogeneous multi-hop network with transmission reliability $p_i^D$, transmission delay $\theta_i$, and feedback delay $\omega_i$. The delayed feedback informs the BS whether or not the packet was successfully received by the destination.}
    \label{fig:Systemmodel}
    \vspace{-1.5\baselineskip}
\end{figure}

\vspace{0.5ex}
\noindent\textbf{Challenge.} Timestamps are generated at the %Application layer of the 
sources and decoded/decrypted at the %Application layer of the 
destinations. %Timestamps may be encrypted at the source, and the BS may or may not have 
Knowledge of the ``\emph{destination timestamps}'' $\tau^D(t)$, namely timestamps of the packets successfully received at the destinations, allows the scheduling policy to know which destinations have outdated information, i.e., large values of AoI. 
In addition, knowledge of the ``\emph{source timestamps}'', namely timestamps of the packets stored at the sources, allows the scheduling policy to know which scheduling decisions can result in the largest AoI reduction upon successful delivery. 
Several works in the literature, including
~\cite{kadota2019scheduling,AoIDQN,li2024aequitas,ji2024age,kadota2019minimizing,zakeri2023minimizing,Whittle_Zhifeng,Whittle_YuPin,kadota2018scheduling,liu2023optimizing,AoI5G,liu2024optimizing,tang2020minimizing}, which are further described below, developed transmission scheduling policies that leverage real-time knowledge of source and/or destination timestamps in making scheduling decisions to minimize AoI. 
These works demonstrate the significant potential benefits of AoI-aware scheduling policies. 
%However, the assumption of perfect (i.e., reliable and without delay) knowledge of the current timestamps at the sources and/or destinations severely limits the scope of their contributions. 
%Perfect knowledge requires sources and/or destinations to continuously share timestamp information with the BS, which imposes that the connection between sources, BS, and destinations must be reliable and without delay. 
However, the assumption of perfect (i.e., reliable and without delay) knowledge of the current timestamps imposes that sources and/or destinations must be able to continuously share timestamp information with the BS which, in turn, imposes that the connection between sources, BS, and destinations must be reliable and without delay, thereby severely limiting the scope of their contributions. 
%For the sources and destinations to be able to share timestamp information with the BS in real-time, they must all be co-located. 
%For the BS to be able to collect timestamp information reliably and without delay, sources and/or destinations must be all co-located. 
%In~\cite{AoIbasicImplementation,WiFresh,WiSwarm}, the authors implemented AoI-aware scheduling policies using real hardware. A common feature of~\cite{AoIbasicImplementation,WiFresh,WiSwarm} is that the sources, BS, and destinations are all co-located as part of the same single-hop wireless network. The quasi-instantaneous feedback from single-hop wireless transmissions was used to determine whether the destination timestamps were updated. %or (at most) one-hop away, allowing the transmission scheduling policy to collect timestamp information.
In general, if the destinations -- e.g., remote monitors or cloud computing systems -- are not co-located with the BS and their connection is unreliable and/or delayed, then the BS would not have perfect knowledge of the timestamps and AoI. Our numerical results in Sec.~\ref{sec:simulation} suggest that even modest communication delays (of a few packet transmission times) can severely affect the performance of AoI-aware scheduling algorithms. 

%$\tau^D(t)$ and, thus, would not know the current value of AoI at the destinations. 

% However, in general, these timestamps are not readily available for the scheduling policy running at the %Data Link layer of the wireless 
% base station. 
% For example, the %Application layer 
% timestamps (and data) generated at the sources could be encrypted, or the destinations (such as remote monitors and cloud computing systems) could be far away from the base station, with a connection that may be unreliable and/or delayed. %, preventing the destinations from keeping 
% In such cases, it is unrealistic to assume that the transmission scheduling policy has perfect knowledge of timestamps at the sources/destinations, and can leverage this information to make scheduling decisions. 

% %network control algorithm (running at the base station) has perfect 

\vspace{0.5ex}
\noindent\textbf{Main Contributions.} In this paper, we consider a network in which the BS has either no knowledge or 
imperfect (i.e., delayed and unreliable) knowledge of source/destination timestamps and AoI. 
Furthermore, we consider sources that generate packets according to a general renewal process, which makes estimating timestamps and AoI more challenging when compared to sources that generate packets deterministically (e.g., immediately before transmission or periodically) or according to Bernoulli processes. Under this general and challenging network model, our main contributions can be summarized as follows:
%a general network with a general packet generation process, a multi-user uplink scenario with both unreliable channels from the sources to the BS and from the BS to the destination with transmission delay, no knowledge of system time, and delayed knowledge of AoI. Therefore, these related scheduling policies\cite{liu2023optimizing},\cite{liu2024optimizing},\cite{ji2024age} are not feasible in our model due to the general packet generation process, more realistic channels, and more limited knowledge. Our main contributions are summarized as follows:
\begin{itemize}[leftmargin=0.15in]
%\item We derive a lower bound on the AoI achievable by any scheduling policy. %To the best of our knowledge, Theorem~\ref{theo:lower_bound} provides the first lower bound for networks with delayed and unreliable connection between the BS and the destinations. %imperfect knowledge of source/destination timestamps and AoI.
\item We derive a closed form-expression for the AoI performance of any  Randomized Policy (see Proposition~\ref{prop:EWSAoIRandomzedproposition}). We derive performance guarantees for the Optimal Randomized Policy with sources that generate packets according to general renewal processes. Remark~\ref{rem:periodic} shows that, from the class of renewal processes with packet generation rate $\lambda_i$, periodic generation achieves the best performance guarantees. 
%Theorem~\ref{} provides the first AoI guarantee for networks with sources that generate packets according to general renewal processes. 
%For the special case of Bernoulli packet generation, Theorem~\ref{} reduces to~\cite[Theorem~?]{kadota2019minimizing}. For the special case of Periodic packet generation, Theorem~\ref{} shows for the first time that the performance
%the AoI expression for the network employing single packet queues with general packet generation processes and develop an optimal randomized policy for any general packet generation process, which can guarantee consistent optimality compared with the lower bound.
\item We propose low-complexity minimum mean square error (MMSE) estimators that use historical information logged by the BS -- including scheduling decisions, transmission outcomes, and delayed timestamps (if available) -- to estimate both the timestamps at the sources and the AoI at the destinations. These MMSE estimators enable AoI-aware scheduling policies in the literature to be used in networks with imperfect knowledge of source/destination timestamps and AoI. %To the best of our knowledge, this is the first MMSE estimator of AoI in the literature. 
%iterative MMSE Estimation algorithm \ref{MWw/Estimationnofeedback} for both AoI and system time in the case of no feedback. 
\item We derive performance guarantees for the well-known AoI-aware Max-Weight Policy~\cite{AoIsurvey,kadota2019minimizing} using the MMSE estimators, as opposed to relying on perfect knowledge of AoI. Theorem~\ref{theo:upperboundMW} provides a performance guarantee for both cases of imperfect knowledge and no knowledge of AoI. To the best of our knowledge, this is the first AoI-aware scheduling policy that attempts to minimize AoI without requiring any AoI knowledge.
%for this policy. % Max-Weight Policy with estimation. 
%The proof of Theorem ~\ref{} supplements the missing steps in the proof of \cite{ji2024age} and improves the performance bound given in \cite{liu2024optimizing}.
\item We evaluate the impact of the packet generation process, imperfect knowledge of timestamps and AoI, and scheduling policy on AoI. Our numerical results show that, contrary to intuition, the performance of the Max-Weight Policy with MMSE estimators is nearly identical for the cases of imperfect knowledge and no knowledge of AoI. 
%, regardless of whether feedback is available. 
%To the best of our knowledge, this is the first AoI-aware scheduling policy that minimizing AoI without requiring any AoI knowledge.
\end{itemize}

\vspace{0.5ex}
\noindent\textbf{Related Work.} The problem of developing transmission scheduling policies that optimize AoI in single-hop wireless networks has been extensively studied in the literature, e.g.,~\cite{kadota2019scheduling,saurav2023scheduling,AoIDQN,li2024aequitas,AoIOFDMA,Vishrant_Multihop,ji2024age,fountoulakis2023scheduling,kadota2019minimizing,Whittle_Zhifeng,Whittle_YuPin,kadota2018scheduling,AoI5G,AoIopenloop,liu2023optimizing,talak2018optimizing,liu2024optimizing,tang2020minimizing,zakeri2023minimizing,chen2024minimizing,ayan2019age}  Sources with different packet generation processes, including packet generation before transmission (also known as generate-at-will)~\cite{kadota2019scheduling,AoIopenloop,AoIDQN,AoIOFDMA,fountoulakis2023scheduling,li2024aequitas,Vishrant_Multihop,ji2024age}, Bernoulli packet generation~\cite{kadota2019minimizing,zakeri2023minimizing,Whittle_Zhifeng,ji2024age,Whittle_YuPin}, and periodic packet generation~\cite{AoI5G,ji2024age}, have been considered. 
Scheduling policies that can select one source per time slot~\cite{zakeri2023minimizing,Whittle_Zhifeng,kadota2018scheduling,kadota2019scheduling,AoIDQN,AoIopenloop,kadota2019minimizing,saurav2023scheduling,liu2023optimizing,ji2024age,Whittle_YuPin} or multiple sources per slot~\cite{liu2024optimizing,AoI5G,AoIOFDMA,tang2020minimizing,li2024aequitas,talak2018optimizing} have also been considered. 
Objective functions focused on minimizing AoI~\cite{kadota2018scheduling,AoIDQN,AoIopenloop,kadota2019minimizing,saurav2023scheduling,talak2018optimizing,AoIOFDMA,AoI5G,li2024aequitas,liu2023optimizing,liu2024optimizing,Vishrant_Multihop,ji2024age,Whittle_YuPin},
% minimizing PAoI~\cite{AoIopenloop}, and
minimizing variants of AoI~\cite{chen2024minimizing,ayan2019age}, and minimizing AoI subject to additional constraints such as energy~\cite{ji2024age,tang2020minimizing,saurav2023scheduling} or throughput~\cite{fountoulakis2023scheduling,kadota2019scheduling} have also been considered. 
Most works~\cite{kadota2019scheduling,AoIDQN,li2024aequitas,kadota2019minimizing,zakeri2023minimizing,Whittle_Zhifeng,Whittle_YuPin,AoI5G,kadota2018scheduling,tang2020minimizing}  assume perfect knowledge of both the timestamps at the sources and the AoI at the destinations. A few recent works~\cite{liu2023optimizing,liu2024optimizing,ji2024age} considered the problem of imperfect knowledge of either timestamps at the sources or AoI at the destinations.

In~\cite{liu2023optimizing,liu2024optimizing}, the authors consider a network with multiple sources sending packets to a BS. Sources generate packets according to a Bernoulli process. The BS schedules multiple sources at every time slot. The BS has imperfect knowledge of the source timestamps, but perfect knowledge of AoI. The authors formulated the problems as a partially observable Markov decision process (POMDP) that captures the ``belief'' of the current source timestamps, proposed an AoI-aware ``dynamic scheduling policy'' with high computational complexity, and then proposed action space reduction to address this challenge. 
In~\cite{ji2024age}, the authors consider a network with a BS generating and sending packets to multiple destinations. The BS generates packets periodically or according to Bernoulli processes. The BS sends a single packet to a selected destination at every time slot. The feedback of the transmission outcome is delayed. The BS has perfect knowledge of the source timestamps, but imperfect knowledge of AoI due to the feedback delay. The authors formulate the problem as a POMDP, develop an AoI estimator based on the perfect knowledge of the source timestamps, and develop a low-complexity ``greedy policy'' that uses the AoI estimator. 

In this paper, we develop estimators (in Sec.~\ref{sec:estimators}) and low-complexity scheduling policies (in Secs.~\ref{sec:random} and~\ref{sec:MW}) that do not rely on a specific packet generation process, e.g., generate-at-will, periodic, or Bernoulli, at the sources and that do not rely on perfect knowledge of source timestamps or AoI at the destinations. These generalized MMSE estimators can be readily used with existing AoI-aware transmission scheduling policies proposed in the literature, including Whittle's Index~\cite{Whittle_YuPin,Whittle_Zhifeng}, Max-AoI First~\cite{kadota2018scheduling,sun2018sources}, and Max-Weight (discussed in Sec.~\ref{sec:MW}).

The remainder of this paper is organized as follows. In Sec.~\ref{sec:system}, we describe the network model. In Sec.~\ref{sec:lower_bound}, we derive a lower bound on the achievable AoI. In Sec.~\ref{sec:random}, we develop and analyze the Optimal Randomized Policy. In Sec.~\ref{sec:estimators}, we propose MMSE estimators for timestamps and AoI. In Sec.~\ref{sec:MW}, we develop and analyze the Max-Weight Policy using the estimators. In Sec.~\ref{sec:simulation}, we provide numerical results. 
% In Sec. VII, we provide the emulation results. 
The paper is concluded in Sec.~\ref{sec:conclusion}.
% Due to the space constraint, some of the technical proofs are provided in the report in [19].

\section{Network Model}\label{sec:system}

Consider a network with $N$ sources sending time-sensitive packets to $N$ destinations through a wireless BS, as illustrated in Fig.~\ref{fig:Systemmodel}.
Let time be slotted, with slot index $t \in \{1, 2, \ldots,$ $ T\}$, where $T$ represents the time-horizon. 
%Denote the slot index by $t \in \{1, 2, \ldots, T\}$, and let the slot duration be normalized to unity. 
Each source generates new packets according to a renewal process. 
Let $a_i(t) = 1$ indicate the event of a packet generation at source $i$ during slot $t$, and $a_i(t) = 0$ otherwise. 
Denote by $X_i$ the \emph{inter-generation period}, i.e., the number of slots between two consecutive packet generation events at source $i$. 
The probability mass function (PMF) associated with $X_i$ is defined as $f_i(x):=\mathbb{P}(X_i = x)$ for 
$x \in\{1,2,\ldots,\bar{x_i}\}$ for $\bar{x_i}<\infty$. 
%$x \in \mathbb{N}$ and $x < \infty$. 
The \emph{average packet generation rate} is defined as $\lambda_i := {1/\mathbb{E}[X_i]}$. 
Since the packet generation process at source $i$ is a renewal process, the inter-generation period $X_i$ are independent and identically distributed (i.i.d.) over packets. Furthermore, we assume that the packet generation process across different sources are independent. 
%At the beginning of each slot $t$, packets are generated by each source $i \in \{1, 2, \ldots, N\}$  according to a renewal process. Let $a_i(t) = 1$ indicate the event of a packet generation at source $i$, and $a_i(t) = 0$ otherwise. Denote by $X_i$ the inter-generation period $X_i$ are the intervals between the slots where $a_i(t) = 1$ and no packets arrive in the intervening slots. The probability mass function (PMF) for $X_i$, denoted $f_i(x) = \mathbb{P}(X_i = x)$ for $x \in \mathbb{N}$ and $x < \infty$. The average packet generation rate $\lambda_i$ is defined as $\lambda_i := {1/\mathbb{E}[X_i]}$. Since $X_i$ models a renewal process, packet generation events are independent and identically distributed (i.i.d.) for each source $i$. Additionally, packet generation processes across different sources are assumed to be independent.

% \begin{figure}[htbp]
%     \centering
%     \includegraphics[width=0.5\textwidth]{Systemmodel.pdf}
%     \caption{Illustration of the network model: the BS serves $N$ sources, and each slot selects up to $K$ sources; the packets in each source $i$ transmitted through the wireless channel with channel reliability $p_i^S$, as well as the multi-hop heterogeneous network channel with channel reliability $p_i^D$, transmission delay $\theta_i$ and feedback delay $\omega_i$.}
%     \label{fig:Systemmodel}
% \end{figure}

Packets generated at source $i$ are enqueued in a \emph{single packet queue} that stores only the latest generated packet. %where a new packet replaces any older packet from the same source. 
The single packet queue is updated
% The Head-of-Line (HoL) packets represent the set of packets from all queues that are available to the BS for transmission.
%For each source \( i \), the packet in the source is updated 
only when a new packet is generated, in which case the new packet replaces any older packet in the queue. 
The single packet queue remains unchanged when its packet is successfully transmitted to the BS. %and no new packet generated. 
This means that: (i) after the first packet generation, the single packet queue will always contain exactly one packet; and (ii) sources can transmit repeated packets to the BS. 
%This means that: (i) after the first packet generation, the single packet queue will never be empty; and (ii) 
Let \( z_i(t) := t - \tau^S_i(t) \) represent the \emph{system time} of the packet stored at source \( i \) in slot \( t \), where \( \tau^S_i(t) \) represents the ``source timestamp'', i.e., the time slot in which the packet stored at source $i$ was generated. %(i.e., the generation time slot) of the packet. 
The system time evolves as:
\begin{equation}\label{zevolves}
    z_{i}(t+1)=\begin{cases} 
0 & \text{if } a_i (t)=1, \\ 
z_{i}(t)+1 & \text{otherwise.} 
\end{cases} 
\end{equation}

At the beginning of slot \( t \), the BS schedules transmissions from up to $K$ sources, where $K \leq N$. %the packets from certain sources. 
Let \( u_i(t) =1 \) indicate that source \( i \) is selected for transmission during slot \( t \), and \( u_i(t) =0 \) otherwise. It follows that %is denoted by the indicator function \( u_i(t) \in \{0,1\} \), where \( u_i(t) = 1 \) indicates that the BS schedules the packet from the source \( i \), and \( u_i(t) = 0 \) otherwise. 
%Due to channel interference constraints, the BS can receive and forward a maximum of \( K \) packets in any given time slot \( t \), hence 
\( \sum_{i=1}^N u_i(t) \leq K ,\forall t\).
Each selected source transmits its packet to the BS via an unreliable wireless channel. 
%The successfully and subsequently, the BS forwards these packets through a multi-hop heterogeneous network (the second channel) to their respective destinations. For the first channel, 
Let \( c^S_i(t)=1 \) indicate that the channel from source \( i \) to the BS is ON during slot \( t \), and \( c^S_i(t)=0 \) otherwise. The source channel $c^S_i(t)$ is i.i.d.\ over time and independent across different sources, with \( \mathbb{P}(c^S_i(t) = 1) = p^S_i \in (0, 1],\forall i,t \). 
The BS does not know $c^S_i(t)$ prior to making scheduling decisions. 
The BS successfully receives a packet from source $i$ in slot $t$ when $u_i(t)c^S_i(t)=1$. 
%The probability that a packet transmitted from source \( i \) is successfully received by the BS is \( \mathbb{P}(c^S_i(t) = 1) = p^S_i \in (0, 1] \). 

%The connection between the BS and the destinations is unreliable and delayed.
Packets that are successfully received by the BS during slot $t$ are immediately forwarded to their corresponding destinations 
%through a heterogeneous multi-hop network with transmission reliability $p_i^D$ and transmission delay $\theta_i$. 
through a heterogeneous multi-hop network with transmission reliability $p_i^D$, transmission delay $\theta_i$, and feedback delay $\omega_i$.
Let \( c^D_i(t)=1 \) indicate that the packet forwarded by the BS during slot \( t \) is successfully received by destination \( i \) in slot \( t + \theta_i \), and \( c^D_i(t)=0 \) otherwise. The destination channel $c^D_i(t)$ is i.i.d. over time and independent across different destinations, with \( \mathbb{P}(c^D_i(t) = 1) = p^D_i \in (0, 1],\forall i,t \). 
Destination $i$ successfully receives a packet from source $i$ in slot $t+\theta_i$ when $u_i(t)c^S_i(t)c^D_i(t)=1$. 
%Upon successful reception of a packet in slot $t+\theta_i$, destination $i$ sends an acknowledgement to the BS. This feedback reaches the BS in slot $t+\theta_i+\omega_i$. The unreliable and delayed connection between BS and destinations limits the knowledge of the BS, as discussed in Sec.~\ref{sec:observation}.

Upon successful reception of a packet in slot $t+\theta_i$, destination $i$ updates \( \tau^D_i(t+\theta_i) \) which represents the ``destination timestamp'', i.e., the time slot in which the latest packet received by destination $i$ was generated. Furthermore, destination $i$ sends an acknowledgment to the BS containing the value of \( \tau^D_i(t+\theta_i) \). This feedback reaches the BS in slot $t+\theta_i+\omega_i$. To study networks with unknown (or time-varying) feedback delay $\omega_i$ and networks with no feedback from the destinations, in Sec.~\ref{sec:estimators}, we consider $\omega_i\rightarrow\infty,\forall i$. 
%Notice that $\omega_i\rightarrow\infty$ unknown and/or time-varying feedback delay 
The unreliable and delayed connection between BS and destinations limits the knowledge of the BS, as discussed in Sec.~\ref{sec:observation}.

%The feedback delay and its implications to the network state observable at the BS are discussed in Sec.~\ref{sec:observation}.
%Upon successful reception of packets from up to $K$ sources, the BS immediately forwards these packets to the corresponding destinations. 
%For the second channel, there is an associated transmission delay \( \theta_i \) for packets transmitted to the destination \( i \), let \( c^D_i(t)\in \{0,1\} \) be an indicator function that equals to 1 when the packet transmitted at slot \( t \) from the BS is successfully received by destination \( i \) at slot \( t + \theta_i \), and 0 otherwise. The probability that a packet from the BS is successfully received by the destination \( i \) is \( \mathbb{P}(c^D_i(t) = 1) = p^D_i \in (0, 1] \).

%While it waits, the base station continues to select sources for transmission in subsequent slots. The feedback from destination~$i$ provides updated information about the latest packet successfully received by the application, allowing the base station to update its belief about the value of AoI $\AoI_i(t)$.

The transmission scheduling policies considered in this paper are non-anticipative, which means that they do not use future information when making scheduling decisions. Let $\Pi$ represent the class of non-anticipative policies and let $\pi \in \Pi$ denote an arbitrary admissible policy. Our goal is to develop scheduling policies $\pi$ %based on the observation $\mathbb{O}(t)$ 
that minimize AoI in the network. Next, we formulate the AoI minimization problem and define the observable network state from the perspective of the BS. %observation.%, respectively.

% Furthermore, \( c^D_i(t) = 0 \) for all \( t < \theta_i \) because packets cannot be transmitted before \(\theta_i\).

\subsection{Age of Information at the Destinations}

The Age of Information quantifies the information freshness from the perspective of the destinations. Let \( h_i(t) := t - \tau^D_i(t) \) be the AoI associated with destination \( i \) at the beginning of slot \( t \). %. By definition, \( h_i(t) := t - \tau^D_i(t) \), 
%where \( \tau^D_i(t) \) represents the generation time of the freshest packet delivered to destination \( i \) before and excluding slot \( t \). 
%If destination \( i \) does not receive a packet from source \( i \) during slot \( t \), then \( h_i(t+1) = h_i(t) + 1 \); if destination \( i \) receives a packet from source \( i \) during slot \( t \), then \( h_i(t+1) = z_i(t-\theta_i) + \theta_i \), since the received packet was transmitted by the source in slot \( t-\theta_i \). 
The evolution of \( h_i(t) \) is given by:
\begin{equation}\label{hevolve}
h_i(t+1) = \begin{cases} 
z_i(t-\theta_i) + \theta_i+1 & \text{if } \Lambda_i(t) = 1, \\ 
h_i(t) + 1 & \text{otherwise,} 
\end{cases}
\end{equation}
where \( \Lambda_i(t) := u_i(t-\theta_i)c^S_i(t-\theta_i)c^D_i(t-\theta_i) \in \{0, 1\}, \forall t > \theta_i, \forall i, \) is an indicator function that is equal to 1 when destination $i$ successfully receives a packet from source $i$ during slot $t$, %if there is a successful delivery from source \( i \) to destination \( i \) during the time slot \( t \), 
and 0 otherwise, with \( \Lambda_i(t) = 0, \forall t \leq \theta_i \) as packets cannot be delivered before \( t=\theta_i \). For simplicity, we assume that \( h_i(1) = 1 \) and \( z_i(1) = 0 ,\forall i \), which means that every source generates a new packet in slot $t=1$. %there is a packet generation at slot 1 at each source \( i \) .

% When a packet is received by destination \( i \), the destination sends an acknowledgement to the BS.
% The acknowledgement reaches the BS after \(\omega_i\) time slots. 
% This acknowledgement provides the BS with information about the latest packet successfully received by the application, 
% allowing the BS to update its knowledge of the packet generation time \( \tau_i^D(t) \) with a feedback delay of \(\omega_i\). In particular, we also consider the case of no-feedback, which is equivalent to \(\omega_i \rightarrow \infty\). 

To capture the information freshness in a network employing a transmission scheduling policy \( \pi \in \Pi \), we define the Expected Weighted Sum AoI (EWSAoI) in the limit as the time horizon $T$ grows to infinity as:
\begin{equation}\label{EWSAoIexpression}
    \mathbb{E}[J^\pi] = \lim_{T \to \infty} \frac{1}{TN} \sum_{t=1}^T \sum_{i=1}^N \alpha_i \mathbb{E}[h_i^{\pi}(t)], 
\end{equation}
where \( \alpha_i >0 \) represents the priority of source \( i \). %Next, we discuss the delayed and unreliable knowledge of the network state from the perspective of the BS. 

\subsection{Observable Network State at the Base Station}\label{sec:observation}

The delayed and unreliable channels between sources, BS, and destinations limits the knowledge of the BS. For example, the BS may not know the current AoI $h_i(t)$ at the destinations and the current system times $z_i(t)$ at the sources. In this section, we define the observable network state from the perspective of the BS. The BS observation $\mathbb{O}(t)$ at the beginning of slot $t$ is fundamental to the development of the MMSE estimator for AoI and system time discussed in Sec.~\ref{sec:estimators}.

At the beginning of slot~$t$, the \emph{BS knows the evolution of the network-wide AoI} \(\{\mathbf{h}(\varphi)\}_{\varphi=1}^{t-\boldsymbol{\omega}}\), where \(\mathbf{h}(\varphi) := \{h_1(\varphi), \ldots, h_N(\varphi)\}\) represents the set of AoI values in slot \(\varphi\); and $\boldsymbol{\omega}=\{\omega_1,$ $\ldots,\omega_N\}$ represents the set of feedback delays. 
%
%For each source \( i \), due to the feedback delay \(\omega_i\), the instantaneous \(\tau^D_i(t)\) becomes uncertain, and the BS must instead rely on the packet generation time of the last packet delivered with a feedback delay \(\omega_i\), that is, \(\tau^D_i(t - \omega_i)\). Furthermore, due to channel interference constraints, the BS cannot obtain instantaneous knowledge of \(\tau^A_i(t)\) at the beginning of the slot \( t \). Therefore, in each time slot \( t \), the observation of the BS from sources and destinations is based on the delayed feedback from each destination \( i \), i.e. \(\tau^D_i(t - \omega_i)\), which is equivalent to obtain \(h_i(t - \omega_i) = t - \tau^D_i(t - \omega_i)\).
%
% When a packet is received by destination \( i \), the destination sends an acknowledgement to the BS.
% The acknowledgement reaches the BS after \(\omega_i\) time slots. 
% This acknowledgement provides the BS with information about the latest packet successfully received by the application, 
% allowing the BS to update its knowledge of the packet generation time \( \tau_i^D(t) \) with a feedback delay of \(\omega_i\). In particular, we also consider the case of no-feedback, which is equivalent to \(\omega_i \rightarrow \infty\). 
%
%In addition to obtaining AoI information from other nodes in the network, the BS can also observe the transmission outcome. 
Furthermore, the \emph{BS knows the history of transmission outcomes} illustrated in Fig.~\ref{fig:Transmission Outcome}. Recall that a packet is successfully received by the BS when $u_i(t)c^S_i(t)=1$. Let $r_i(t)=1$ indicate the reception of a \emph{repeated packet} that had already been successfully received by the BS in a slot prior to $t$, and $r_i(t)=0$ otherwise. Let \( D_i(t) \) be the index of the latest packet received by the BS from source \( i \) at the beginning of slot \( t \), which evolves as follows:
\begin{equation}\label{eq:D_i(t)}
    D_i(t+1) = \begin{cases}
        D_i(t) + 1 & \text{if } u_i(t)c^S_i(t)(1-r_i(t))=1, \\
        D_i(t) & \text{otherwise,}
    \end{cases}
\end{equation}
with $D_i(1)=0,\forall i$. %Notice that the reception of repeated packets does not change $D_i(t)$.
%naturally, repeated packets ($r_i(t)=1$) do not affect the AoI at the destination.
For each packet \( D_i(t) \geq 1 \), denote by \({\mathcal{T}}_i[D_i(t)]\) the set of slots in which successful transmissions of packet \( D_i(t) \) occurred. This set is initialized as $\mathcal{T}_i[d]=\emptyset,$ $\forall d\geq 1$ and evolves as:
\begin{equation}\label{Tevolves}
    {\mathcal{T}}_i[D_i(t+1)] = \begin{cases}
        {\mathcal{T}}_i[D_i(t+1)] \cup \{t\} & \text{ if } u_i(t)c^S_i(t)=1, \\
        {\mathcal{T}}_i[D_i(t+1)] & \text{otherwise.}
    \end{cases}
\end{equation}
Notice that the reception of a repeated packet in slot $t$ does not change $D_i(t+1)$ and increases the cardinality of $|{\mathcal{T}}_i[D_i(t+1)]|$ by one element. 
The first slot in which packet \( D_i(t) \) was successfully received by the BS is denoted by \(\delta_i[D_i(t)] = \min({\mathcal{T}}_i[D_i(t)])\), and the last slot by \(\bar{\delta}_i[D_i(t)] = \max({\mathcal{T}}_i[D_i(t)])\), as illustrated in Fig.~\ref{fig:Transmission Outcome}.
%Then, \(\delta_i[D_i(t)] = \min({\mathcal{T}}_i[D_i(t)])\) represents 
%Define \(\delta_i[D_i(t)] := \min({\mathcal{T}}_i[D_i(t)])\) as 
%the first slot in which packet \( D_i(t) \) was successfully received by the BS, and \(\bar{\delta}_i[D_i(t)] = \max({\mathcal{T}}_i[D_i(t)])\) represents the last slot, as illustrated in Fig.~\ref{fig:Transmission Outcome}. %that packet \( D_i(t) \) was successfully transmitted. 

\begin{figure}[t]
\centering
\includegraphics[width=0.5\textwidth]{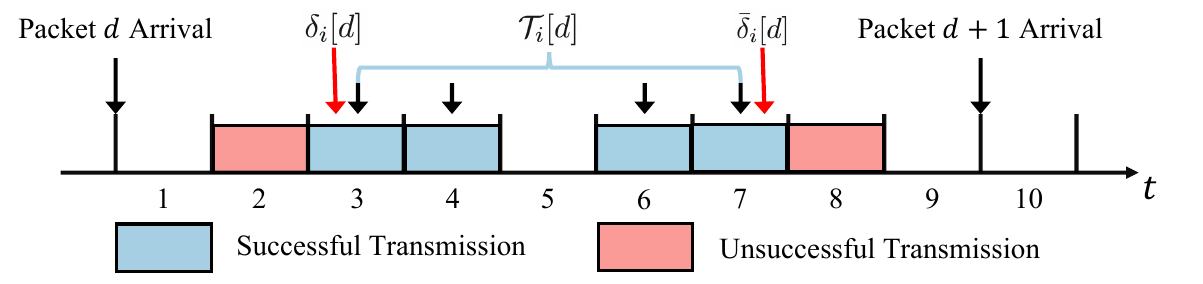} 
\vspace{-1.5\baselineskip}
\caption{Illustration of the history of transmission outcomes. The BS keeps track of the indices of the time slots in which successful packet transmissions occurred. These slots are part of the BS observation $\mathbb{O}(t)$ in~\eqref{eq:observation} that enables the BS to estimate AoI and system times.} %for each scheduled packet that successfully received by the BS in certain slots, record the packet's index along with the indices of each slot in which it was successfully received.}
\label{fig:Transmission Outcome}\vspace{-1.5\baselineskip}
\end{figure}

We define the \emph{BS observation} at the beginning of slot \( t \) as:
\begin{equation}\label{eq:observation}
    \mathbb{O}(t) := \left\{ \{\mathbf{h}(\varphi)\}_{\varphi=1}^{t-\boldsymbol{\omega}}, \mathbb{T}(t) \right\}, %\forall i
\end{equation}
%where \(\{\mathbf{h}(\varphi)\}_{\varphi=1}^{t-\boldsymbol{\omega}}\) is the AoI knowledge history with delay, \(\mathbf{h}(\varphi) := \{h_1(\varphi), h_2(\varphi), \ldots, h_N(\varphi)\}\) represents the network-wide AoI of slot \(\varphi\) and $\boldsymbol{\omega}=\{\omega_1,\omega_2,...,\omega_N\}$ represents the network-wide feedback delay; 
where $\mathbb{T}(t):=\{ {\mathcal{T}}_i[d]\}_{d=1}^{D_i(t)}$ represents the history of transmission outcomes up to and excluding slot~$t$. %; and $\mathbf{u(t)c^S(t)}:=\{u_i(t)c_i^S(t)\}_{i=1}^N$, $\mathbf{r}(t):=\{r_i(t)\}_{i=1}^N$ represent the transmission outcome at slot~$t$.

\begin{remark}
    %The BS observation $\mathbb{O}(t)$ keeps track of the history of successful packet transmissions. 
    Failed transmissions are not accounted for in~\eqref{eq:observation} since they do not provide useful information to estimate the AoI or system times. Transmission failures from a given source~$i$ occur with probability \( 1 - p^S_i \) irrespective of the packet being transmitted. Furthermore, failed transmissions are not decoded or forwarded. Hence, they do not provide insight into which packet was transmitted nor its generation time. 
\end{remark}

\section{Lower Bound on the achievable EWSAoI}\label{sec:lower_bound}
In this section, we derive a lower bound on the achievable EWSAoI by any admissible scheduling policy $\pi \in \Pi$. The lower bound in Theorem~\ref{theo:lower_bound} generalizes the result in~\cite{kadota2019minimizing} to the more challenging network model with: (i)~delayed and unreliable connection between the BS and the destinations; and (ii)~transmission scheduling policies that can select multiple sources in each slot~$t$. 

\begin{theorem}\label{theo:lower_bound}
For any given network model with parameters $\{N,K,\alpha_i,p_i^S,p_i^D,\lambda_i,\theta_i\}$ and any renewal packet generation processes at the sources with PMF $f_i(x)=\mathbb{P}(X_i = x)$, 
%For any given network model with parameters $\{N,K,\alpha_i,p_i^S,p_i^D,\lambda_i\}$ and any renewal packet generation processes at the sources, %the optimization problem in (\ref{eq:lowerboundobj}) -- (\ref{eq:constraint2forLB}) provides a 
the lower bound on the AoI minimization problem in~\eqref{EWSAoIexpression}, namely $L_B \leq \mathbb{E}[J^\pi],\forall \pi\in\Pi$, is given by
\begin{equation}\label{performancelowerbound}
 L_B = \frac{1}{2N} \sum_{i=1}^N \alpha_i \left( \frac{1}{q_i^{L_B}}  + 2\theta_i +1 \right)    
\end{equation}
where $q_i^{L_B}$ yields from Algorithm 1.
% \begin{equation}
%     q_i^{L_B} = v_i \min\{1; \sqrt{\frac{\gamma_i}{\gamma^*}  }\}, \forall i\; ,
% \end{equation}
% $v_i = \min\{\lambda_i, p^S_i p^D_i\}$, $\gamma_i = \alpha_i p^S_ip^D_i/ (2N v_i^2), \forall i,$ and $\gamma^*$ yields from Algorithm 1.
\end{theorem}

\begin{proof}
The solution of the optimization problem in~\cite[Eqs.~(14a)-(14c)]{kadota2019minimizing} gives a lower bound to a simplified network model that assumes perfect knowledge of AoI and a single transmission per slot. We generalize~\cite[Eqs.~(14a)-(14c)]{kadota2019minimizing} to the more practical network model in Sec.~\ref{sec:system}. Algorithm~\ref{alg:lowerbound} finds the unique solution to the Karush–Kuhn–Tucker (KKT) conditions of the generalized optimization problem. The details have been omitted because of space constraints.
% For (\ref{eq:lowerboundobj}) -- (\ref{eq:constraint2forLB}), we constructing the Lagrangian function and considering the complementary slackness conditions and dual feasibility, then
% % we identified two scenarios for long-term throughput: 1) sources with positive dual variables have throughput equal to the maximum value; 2) sources with zero dual variables have throughput determined by the parameter $\gamma$. 
% by iteratively adjusting $\gamma$ to satisfy the constraints, we obtained the unique solution $\gamma^*$ from Algorithm 1. The details has been omitted because of space constraints.
\end{proof}

\begin{algorithm}
\caption{Solution to the Lower Bound}\label{alg:lowerbound}
\begin{algorithmic}[1]
\STATE $\tilde{\gamma} = \left(\sum_{i=1}^N \sqrt{\alpha_i / p^S_ip^D_i}\right)^2 / (2NK^2)$
\STATE $v_i = \min\{\lambda_i,p^S_ip^D_i\}$ and $\gamma_i = \alpha_i p^S_ip^D_i/ (2N v_i^2), \forall i$
\STATE $\gamma = \max\{\tilde{\gamma}; \gamma_i\}$
\STATE $q_i = v_i \min\{1; \sqrt{\gamma_i / \gamma}\}, \forall i$, and $S = \sum_{i=1}^N q_i / p^S_ip^D_i$
\WHILE{$S < K$ and $\gamma > 0$}
    \STATE decrease $\gamma$ slightly
    \STATE repeat step 4 to update $q_i$ and $S$
\ENDWHILE
\RETURN $q_i^{L_B} = q_i, \forall i$
\end{algorithmic}
\end{algorithm}

\section{Randomized Scheduling Policy}\label{sec:random}

Denote by $\Pi_r$ the class of Randomized policies. 
Let $r \in \Pi_r$ be a scheduling policy that, in each slot~$t$, selects source $i$ with probability $\mu^r_i \in (0,1]$. %It follows that $\sum_{i=1}^N \mu_i = K - \mu_0$ due to the channel interference constraints. 
If source $i$ is selected during slot $t$, then it transmits a packet to the BS. If this packet is successfully received, the BS immediately forwards this packet to destination $i$. %to destination $i$, and $u_i(t) = 0$ otherwise.
The scheduling probabilities $\mu_i$ are fixed over time and satisfy $\sum_{i=1}^N \mu^r_i = K$. %Without loss of optimality, we assume that exactly $K$ sources are selected in every time slot, which could be achieved by mapping the scheduling probabilities of sources to the scheduling probabilities of $\left(^N_K\right)$ $K$-sparse binary \emph{possible service vectors}~\cite[Chapter~2.3]{neely2022stochastic} of length $N$. %due to the channel interference constraints. 
Each randomized policy $r$ is fully characterized by the set $\{\mu^r_i\}_{i=1}^N$. These policies select sources at random\footnote{Consider the collection of $\left(^N_K\right)$ subsets of $K$ out of $N$ sources. Each subset is associated with a scheduling probability. To select exactly $K$ sources in each time slot, policy $r$ selects a subset according to this probability vector. It can be shown that $\mu^r_i\in[0,1],\forall i,$ and $\sum_{i=1}^N \mu^r_i = K$ are necessary and sufficient to represent any probability vector.}, without taking into account the current AoI $h_i(t)$, system time $z_i(t)$, nor the transmission delay $\theta_i$. 

In this section, we obtain the Optimal Randomized Policy and provide performance guarantees in terms of AoI. 
Specifically, Proposition~\ref{prop:EWSAoIRandomzedproposition} provides the EWSAoI associated with an arbitrary randomized policy $r \in \Pi_r$ and Theorem~\ref{theo:stationary_performance} characterizes the Optimal Randomized Policy and its performance guarantee. %as a solution to an optimization problem. 
%Then, we formulate the Optimal Randomized Policy for any packet generation processes and prove the performance guarantee.

%\subsection{EWSAoI for Randomized Policy}
%Consider a network with $N$ sources and destinations, where each source $i$ has a packet generation process following $f_i(x)$, priorities $\alpha_i$, channel reliabilities $p^S_i$ and $p^D_i$, and a transmission delay $\theta_i$ from the BS to the destination $i$. The BS selects sources according to $r \in \Pi_r$, with scheduling probabilities $\{\mu^r_i\}_{i=1}^N$. We derive following proposition of the achievable EWSAoI for network employing arbitrary randomized scheduling policy $r$. 
% Consequently, as a by-product, the EWSAoI of any network employing LIFO G/Ber/1  also conforms to Proposition \ref{prop:EWSAoIRandomzedproposition}.

\begin{proposition}\label{prop:EWSAoIRandomzedproposition}
For any given network model with parameters $\{N,K,\alpha_i,p_i^S,p_i^D,\lambda_i,\theta_i\}$, any renewal packet generation processes at the sources with PMF $f_i(x)=\mathbb{P}(X_i = x)$, and an 
%For any given network model with parameters $\{N, K, \alpha_i, p_i^S, p_i^D, \lambda_i,f_i, \theta_i\}$, for an 
arbitrary randomized policy $r \in \Pi_r$ with scheduling probabilities $\{\mu^r_i\}_{i=1}^N$, the EWSAoI is given by:
\begin{equation}\label{EWSAoIforG/Ber/1}
    \mathbb{E}[J^r] = \frac{1}{N}  \sum_{i=1}^N \alpha_i \left(\frac{\mathbb{E}\left[X_i^{2}\right]\lambda_i}{2} + \frac{1}{p^D_i p_i^S \mu^r_i} + \theta_i - 1\right)
\end{equation}
\end{proposition}
\begin{proof}
The average AoI for a Last-In First-Out (LIFO) G/G/1 queue is derived in~\cite[Theorem~3]{tripathi2019age}. Notice that the average AoI for a LIFO queue is equivalent to the average AoI for the single packet queue described in Sec.~\ref{sec:system}. Furthermore, for any randomized policy $r \in \Pi_r$, the 
%We adopt the expression for the average AoI for a LIFO G/G/1 queue as given by \cite{tripathi2019age}, which is equivalent to single packet queues we consider.
%Notice that the 
service time of a packet %time of each delivered packet to the delivered time of the packet, 
is geometrically distributed with parameter \(p^S_ip^D_i\mu^r_i\). Substituting the service time and the inter-generation period $X_i$ in~\cite[Theorem~3]{tripathi2019age}, and accounting for the delay $\theta_i$, we get~\eqref{EWSAoIforG/Ber/1}. 
%Consequently, as a by-product, the EWSAoI of any network employing LIFO G/Ber/1  also conforms to Proposition \ref{prop:EWSAoIRandomzedproposition}.
\end{proof}

%\subsection{Optimal Randomized Policy}
From the expression for the average AoI obtained in~\eqref{EWSAoIforG/Ber/1}, we can find the Optimal Randomized Policy with scheduling probabilities $\{\mu^R_i\}_{i=1}^N$ given by the optimization problem below: %is to select optimal scheduling probabilities that minimize the \eqref{EWSAoIforG/Ber/1} is equivalent to the following optimization problem:
\begin{align}
    \min_{r\in \Pi_r}   &\sum_{i=1}^N \frac{\alpha_i }{p^D_i p_i^S \mu^r_i}, \text{s.t.} \sum_{i=1}^N \mu^r_i = K \mbox{ and } 0 < \mu^r_i \leq 1, \forall i.\label{eq:OPTRandom}
\end{align}
Theorem~\ref{theo:stationary_performance} provides the solution to~\eqref{eq:OPTRandom}. 
Interestingly, \emph{the objective function and the optimal scheduling probabilities are independent of the packet generation processes}. %, meaning the optimal scheduling probabilities remain the same for a given network, regardless of the packet generation process. Next, we solve the optimization problem in (\ref{OPT for Randomized})-(\ref{eq:OPTRandom}) to obtain the optimal scheduling probabilities and prove the optimality of the randomized policy.

\begin{theorem}\label{theo:stationary_performance}
For any given network model with parameters $\{N,K,\alpha_i,p_i^S,p_i^D,\theta_i\}$ and any renewal packet generation processes, the scheduling probabilities associated with the Optimal  Randomized policy $R$ are given by:
\begin{equation}\label{eq:solutionRandomized}
    \mu^R_i = \min\left\{1; \sqrt{\frac{\vartheta_i}{\vartheta^*}  }\right\}, \forall i
\end{equation}
where $\vartheta_i = \alpha_i/p^D_ip_i^S , \forall i$, and $\vartheta^*$ yields from Algorithm 2. 
Substituting $\mu^R_i$ from~\eqref{eq:solutionRandomized} into the EWSAoI in~\eqref{EWSAoIforG/Ber/1} gives the performance of the Optimal Randomized Policy $R$. % which  is given by \eqref{EWSAoIforG/Ber/1}, where $\mu^r_i$ is given by~\eqref{eq:solutionRandomized}
% \begin{equation}\label{EWSAoIforRandomized}
%     \mathbb{E}[J^R] = \frac{1}{N}  \sum_{i=1}^N \alpha_i \left(\frac{\mathbb{E}\left[X_i^{2}\right]\lambda_i}{2}+\frac{1}{p^D_ip_i^S\mu^R_i}+ \theta_i -2\right) \; ,
% \end{equation}
%and related to the lower bound according to
It follows that the performance of $R$ is guaranteed to be within:
\begin{equation}\label{eq:Randomized_guarantee}
    L_B\leq\mathbb{E}[J^R] \leq \rho L_B
\end{equation}
where $\rho$ represents the optimality ratio of $R$, and is given by
\begin{equation}\label{eq:optimality}
    \rho = \frac{\sum_{i=1}^N\alpha_i\mathbb{E}\left[X_i^{2}\right]\lambda_i^2}{\sum_{i=1}^N\alpha_i}+2
\end{equation}
%For Bernoulli packet generation processes, $\rho=4$; for Periodic packet generation processes, $\rho=3$, achieving the minimum optimality \(\rho\).
\end{theorem}
\begin{proof}
The proof of (\ref{eq:solutionRandomized}) is similar to the proof of Theorem \ref{theo:lower_bound}. Algorithm~\ref{alg:Randomized} finds the unique solution to the KKT conditions associated with~\eqref{eq:OPTRandom}.
The proof of \eqref{eq:Randomized_guarantee} and \eqref{eq:optimality} are discussed below.
Since $\{\mu_i^R\}_{i=1}^N$ minimizes (\ref{EWSAoIforG/Ber/1}), it follows that any randomized policy $r \in \Pi_r$ is such that $\mathbb{E}\left[{J}^r\right]\geq\mathbb{E}\left[{J}^R\right]$. Consider $\tilde{R}$ where $\tilde{\mu}_i = q_i^{L_B}/p^S_ip^D_i$. Substituting $\tilde{\mu}_i$ into the RHS of (\ref{EWSAoIforG/Ber/1}), we obtain:
\begin{equation} \label{performancerandomized}
        \mathbb{E}[J^{\tilde{R}}] = \frac{1}{N}  \sum_{i=1}^N \alpha_i \left(\frac{\mathbb{E}\left[X_i^{2}\right]\lambda_i}{2}+\frac{1}{q_i^{LB}}+\theta_i-1\right)
\end{equation}
%Let $\rho$ represent the optimality of Optimal Randomized Policy, i.e. $\mathbb{E}[J^R]\leq\rho L_B$, 
By comparing (\ref{performancelowerbound}) and (\ref{performancerandomized}), and defining $\rho$ as in~\eqref{eq:optimality} we establish~\eqref{eq:Randomized_guarantee}. 
% \begin{equation}\label{eq:optimality}
%     \rho = \frac{\sum_{i=1}^N\alpha_i\mathbb{E}\left[X_i^{2}\right]\lambda_i^2}{\sum_{i=1}^N\alpha_i}+2
% \end{equation}
% Substituting the corresponding PMF into \eqref{eq:optimality}, we find that the Optimal Randomized Policy is 4-optimal for Bernoulli packet generation processes and 3-optimal for Periodic packet generation processes. %Applying Jensen's inequality, we demonstrate that the minimum optimality \(\rho\) is 3. 
\end{proof}

\begin{algorithm}
\caption{Optimal Randomized Policy}\label{alg:Randomized}
\begin{algorithmic}[1]
\STATE $\tilde{\vartheta} = \sum_{i=1}^N \sqrt{p^D_i p_i^S /\alpha_i}/K^2$
and $\vartheta_i = \alpha_i/p^D_ip_i^S , \forall i$
\STATE $\vartheta =  \max\{\tilde{\vartheta}; \vartheta_i\}$
\STATE $\mu_i = \min\{1; \sqrt{\vartheta_i / \vartheta}\}, \forall i$
and $S = \sum_{i=1}^N \mu_i$
\WHILE{$S < K$ and $\vartheta > 0$}
    \STATE decrease $\vartheta$ slightly
    \STATE repeat step 3 to update $\mu_i$ and $S$
\ENDWHILE
\RETURN 
% $\vartheta^* = \vartheta$ and 
$\mu_i^{R} =\mu_i, \forall i$
\end{algorithmic}
\end{algorithm}%\vspace{-1.5\baselineskip}

% \begin{remark}
%     From Proposition~\ref{prop:EWSAoIRandomzedproposition}, we can see that the objective function in~\eqref{} and the optimal scheduling probabilities are independent of the packet generation processes. 
% \end{remark}
\begin{remark}\label{rem:periodic}
    The performance guarantee $\rho$ in~\eqref{eq:optimality} provided in Theorem~\ref{theo:stationary_performance} is valid for general renewal packet generation processes at the sources. 
    By applying Jensen's inequality to~\eqref{eq:optimality}, we can see that the minimum optimality \(\rho\) is 3.
    Furthermore, we can see that Bernoulli packet generation result in $\rho=4$ and periodic packet generation result in $\rho=3$.
    %Furthermore, applying Jensen's inequality to~\eqref{eq:optimality}, we can see that the minimum optimality \(\rho\) is 3. 
    Hence, we show that, from the class of renewal processes with packet generation rate $\lambda_i$, periodic packet generation achieves the best performance guarantee $\rho$ in terms of EWSAoI. 
    This result agrees with prior work on continuous-time queueing systems~\cite{AoIM/M/1,AoI_general} which showed that deterministic packet generation leads to a reduction in AoI. %[improve...]    
\end{remark}

\section{MMSE Estimation of AoI and System Times}\label{sec:estimators}

Several works in the literature (e.g.,~\cite{kadota2019scheduling,AoIDQN,li2024aequitas,ji2024age,kadota2019minimizing,zakeri2023minimizing,Whittle_Zhifeng,Whittle_YuPin,kadota2018scheduling,liu2023optimizing,AoI5G,liu2024optimizing,tang2020minimizing}) developed AoI-aware transmission scheduling policies that assumed that the BS had perfect (i.e., reliable and without delay) knowledge of system times at the sources and/or AoI at the destinations. In this section, we consider a network (described in Sec.~\ref{sec:system}) in which the BS has imperfect knowledge of system times and AoI. To employ AoI-aware scheduling policies in such a network, we develop MMSE estimators. In Sec.~\ref{sec:MW}, we apply these MMSE estimators to an AoI-based Max-Weight Policy and analyze its performance. %that can be used in AoI-aware scheduling policies such as~\cite{kadota2019scheduling,AoIDQN,li2024aequitas,ji2024age,kadota2019minimizing,zakeri2023minimizing,Whittle_Zhifeng,Whittle_YuPin,kadota2018scheduling,liu2023optimizing,AoI5G,liu2024optimizing,tang2020minimizing}.  

Scheduling decisions made by the BS at the beginning of slot $t$ result in transmissions of packets with system times $z_i(t)$ which, in turn, may affect the evolution of AoI in time slot $t+\theta_i+1$. According to~\eqref{hevolve}, the value of $h_i(t+\theta_i+1)$ can depend on both $z_i(t) = t - \tau^S_i(t)$ and $h_i(t+\theta_i)=t+\theta_i-\tau^D_i(t+\theta_i)$. Since at the beginning of slot $t$, the BS does not know $\tau^S_i(t)$ nor $\tau^D_i(t+\theta_i)$, we develop MMSE estimators for both quantities. Prior to developing the MMSE estimators, we derive an auxiliary algorithm that iteratively calculates the PMF for the generation time of the latest packet generated before a given slot \(\varphi \leq t\), based on the BS observation \(\mathbb{O}(t)\) defined in~\eqref{eq:observation}. Notice that this PMF is important to estimate both timestamps $\tau^S_i(t)$ and $\tau^D_i(t+\theta_i)$.

\subsection{PMF of Packet Generation}
Let $g_i\left(\varphi', \varphi |\mathbb{O}(t)\right)$ represent the probability that the latest packet generated by source $i$ before (and including) slot $\varphi$ was generated during slot $\varphi'$, given the BS observation $\mathbb{O}(t)$. 
Recall that $\delta_i[d]$ represents the first slot in which the packet with index \( d \) was successfully received by the BS. It follows that $g_i\left(\varphi', \delta_i[d] |\mathbb{O}(t)\right) $ represents the PMF of the generation time of each packet $d\in\{1,\ldots,D_i(t)\}$. Furthermore, $g_i\left(\varphi', t |\mathbb{O}(t)\right) $ represents the PMF of ${\tau}^{{S}}_i(t)$. Next, we derive an expression for the probability of a packet generation in a given slot $\varphi$ and an expression for $g_i\left(\varphi', \varphi |\mathbb{O}(t)\right)$. Algorithm~\ref{alg:arrivalalgorithmnofeedback} describes how to iteratively update both expressions over time. 

Consider the sequence of packets $d\in\{1,\ldots,D_i(t)\}$ from source $i$ successfully received by the BS up to and excluding slot $t$. From the definition of $\delta_i[d]$ and $\bar{\delta}_i[d]$ (see Fig.~\ref{fig:Transmission Outcome} for an illustration) we know that packet $d$ was generated in the interval $\left[\bar{\delta}_i[d-1]+1,\delta_i[d]\right]$, without loss of generality, we define $\bar{\delta}_i[0]=0, \forall i$. Furthermore, we know that the undelivered packet (if any) stored at source $i$ in the beginning of slot $t$ was generated in the interval $[\bar{\delta}[D_i(t)]+1,t]$. We denote by $\cup^+:=\bigcup_{d=1}^{D_i(t)-1}[\bar{\delta}[d]+1,{\delta}[d+1]]\bigcup[\bar{\delta}[D_i(t)]+1,t]$ the set of intervals in which packets can be generated.
%
% Recall that the packet generation process is a renewal process with inter-generation periods $X_i$ distributed according to the PMF $f_i(x)$. %Since the packet generation process is a renewal process, 
% Therefore, if a packet were generated in slot $\varphi'$, the probability of a packet generation in slot $t>\varphi'$ is independent of packet generations before slot $\varphi'$.  
%
% What I'm trying to convey is that $\lambda_i(\varphi)$ is determined by $\varphi,\varphi',g(\varphi',\varphi|\mathbb{O}(t))$. However we did not have $g(\varphi',\varphi|\mathbb{O}(t))$ before generating $\lambda_i(\varphi)$, so we need to use $g(\varphi',\varphi|\mathbb{O}(t))$ and $g(\varphi|\mathbb{O}(t -1))$ and $\cup^+$ to jointly denote $g(\varphi',\varphi|\\mathbb{O}(t))$.(but not for each \varphi) Then we generate $\lambda_i(\varphi)$ and $g(\varphi',\varphi|\mathbb{O}(t))$ sequentially.
%
We denote by $\lambda_i(\varphi):=\mathbb{P}\left(a_i(\varphi)=1|\mathbb{O}(t)\right)$ the conditional probability of a packet generation in slot $\varphi$. For simplicity of notation, the BS observation $\mathbb{O}(t)$ is omitted in $\lambda_i(\varphi)$. 
Naturally, $\lambda_i(\varphi)=0$ if $\varphi\notin\cup^+$. 
%There is no packet generated during $\left[\delta_i[d]+1,\bar{\delta}_i[d]\right]$, otherwise the source transmits a fresher packet at slot $\bar{\delta}_i[d]$. Therefore, $\cup^+:=\bigcup_{d=1}^{D_i(t)-1}[\bar{\delta}[d]+1,{\delta}[d+1]]\bigcup[\bar{\delta}[D_i(t)]+1,t]$ is the sample space of $\varphi$. 
The expression for $\lambda_i(\varphi)$ is as follows  
\begin{equation}\label{lambda_inofeedback}
    \lambda_i(\varphi)= \frac{\sum_{\varphi'=1}^{\varphi-1}g_i(\varphi',\varphi-1|\mathbb{O}(t-1))f_i(\varphi-\varphi')\mathbb{I}_{\varphi\in \cup^+}}{\sum_{\varphi'=1}^{\varphi-1}g_i(\varphi',\varphi-1|\mathbb{O}(t-1))\mathbb{I}_{\varphi\in \cup^+}},
\end{equation}
Recall that we assumed that every source generates a packet in slot $t=1$, hence $\lambda_i(1)=1, \forall i$. 

By definition, $g_i(\varphi',\varphi|\mathbb{O}(t))$ represents the probability of a packet generation in slot $\varphi'$ and no packet generation in the interval $(\varphi',\varphi]$. Notice that $g_i(\varphi',\varphi|\mathbb{O}(t))$ can only be non-zero when $\varphi'\in[\bar{\delta}[D_i(\varphi)-1]+1,\varphi]$. The expression for $g_i(\varphi',\varphi|\mathbb{O}(t))$ is given by %is the sample space of $\varphi'$, $g_i(\varphi',\varphi|\mathbb{O}(t))$ can be expressed as:
\begin{equation}\label{g^S_inofeedback}
g_i(\varphi',\varphi|\mathbb{O}(t))= \frac{
\lambda_i(\varphi')\prod_{\eta=\varphi'+1}^{\varphi}(1-\lambda_i(\eta))\mathbb{I}_{\varphi'>\bar{\delta}[D_i(\varphi)-1]}}{1-\prod_{\eta=\varphi'+1}^{\varphi}(1-\lambda_i(\eta))\mathbb{I}_{\varphi'>\bar{\delta}[D_i(\varphi)-1]}}    
\end{equation}

% Utilizing historic PMF of last generation time $\{g_i\left(\varphi',\varphi|\mathbb{O}(t-1)\right)\}_{\varphi'=\bar{\delta}_i[D_i(t)-1]+1}^{\varphi-1}$ and the observation $\mathbb{O}(t)$, we prove following Algorithm \ref{arrivalalgorithm}, which allows us to iteratively obtain $\lambda_i(\varphi)$  and $g_i\left(\varphi',\varphi|\mathbb{O}(t)\right)$.

Given slots $t$ and $\varphi$, such that $\varphi\leq t$, the probability $\lambda_i(\varphi)$ depends on the BS observation $\mathbb{O}(t)$ and the PMF $g_i(\varphi',\varphi-1|\mathbb{O}(t-1))$ associated with slots $\varphi'\in[\bar{\delta}_i[D_i(\varphi)-1]+1,\varphi-1]$.
In turn, 
%Notice that for each slot $\varphi$, the probability $\lambda_i(\varphi)$ depends on historic PMF $\{g_i(\varphi',\varphi|\mathbb{O}(t-1))\}_{\varphi'=\bar{\delta}_i[D_i(\varphi)-1]+1}^{\varphi-1}$ and the observation $\mathbb{O}(t)$; 
the probability $g_i(\varphi',\varphi|\mathbb{O}(t))$ depends on the PMF $\lambda_i(\eta)$ for slots $\eta\in[\bar{\delta}_i[D_i(\varphi)-1]+1,\varphi]$. Therefore, both probabilities can be iteratively obtained %the PMF of last generation time $g_i(\varphi',\varphi|\mathbb{O}(t))$ can be iteratively obtain by historic PMF and observation $\mathbb{O}(t)$, 
by following the steps outlined in Algorithm \ref{alg:arrivalalgorithmnofeedback}.

\begin{algorithm}
\caption{Iterative update of $\lambda_i(\varphi)$ and $g_i(\varphi',\varphi|\mathbb{O}(t))$}
\label{alg:arrivalalgorithmnofeedback}
\begin{algorithmic}[1]
\STATE {\bf Input:} Observation: $\mathbb{O}(t)$; Historical PMF: $g_i\left(\varphi', \varphi |\mathbb{O}(t-1)\right), \forall \varphi \leq t-1, \varphi' \leq \varphi$; 
% \STATE Firstly, replace the estimation value with the generation time based on the delayed feedback.
% \IF {$\tau^D_i(t-\omega_i)>\bar{\delta}_i[D_i(t-1)]$}
% \FOR {$\varphi'<\bar{\delta}_i[D_i(t)]$}
% \STATE $g_i(\varphi',\varphi,t-1)=\mathbb{I}_{\varphi'=\tau^D_i(t-\omega_i)}(1-\sum_{\eta=\bar{\delta}_i[D_i(t)]}^tg_i(\varphi',\varphi,t-1))$
% \ENDFOR
% \ENDIF
% \STATE Utilizing updated $g_i(\varphi',\varphi,t-1)$, $f_i$, and $\mathbb{O}(t)$ Compute $\lambda_i(\varphi)$ and $g_i(\varphi', \varphi,t)$ 
\FOR{$\varphi = \bar{\delta}_i[D_i(t)] + 1$ to $t$}
\STATE Compute $\lambda_i(\varphi)$ using (\ref{lambda_inofeedback})
\FOR{$\varphi' = \bar{\delta}_i[D_i(\varphi)-1] + 1$ to $\varphi$}
\STATE Compute $g_i(\varphi',\varphi|\mathbb{O}(t))$ using (\ref{g^S_inofeedback})
\ENDFOR
\STATE Store $\lambda_i(\varphi)$ and $g_i(\varphi', \varphi | \mathbb{O}(t))$ for all $\varphi' \leq \varphi$
\ENDFOR
\STATE {\bf Output:}  $g_i(\varphi', \varphi | \mathbb{O}(t))$ for all $\varphi \leq t, \varphi' \leq \varphi$
\end{algorithmic}
\end{algorithm}

Notice that Algorithm~\ref{alg:arrivalalgorithmnofeedback} does not take into account the delayed feedback received from the destinations. %In the case of finite feedback delay, 
A feedback reception in slot~$t$ would provide information about the updated destination timestamp \(\tau^D_i(t-\omega_i)\) which implies that \(\lambda_i\left(\tau^D_i(t-\omega_i)\right)=1\). 
This feedback would affect the probabilities \(\lambda_i(\varphi)\) and \(g_i(\varphi', \varphi|\mathbb{O}(t))\) in every slot \(\varphi \in [\tau^D_i(t-\omega_i), t]\). 
Therefore, upon receiving feedback from a destination, the BS employs Algorithm \ref{alg:arrivalalgorithmnofeedback} to recompute \(\lambda_i(\varphi)\) and \(g_i(\varphi', \varphi|\mathbb{O}(t))\) starting from \(\varphi=\tau^D_i(t-\omega_i)\).

\subsection{MMSE Estimators}
% The estimated value of $\tau^D_i(t+\theta_i)$ is jointly determined by two types of observation: delayed feedback and transmission outcome. In this subsection, we first analyze the MMSE estimation given the transmission outcome \(\mathbb{T}_i(t), \mathbf{c}^S_i(t)\), and \(\mathbf{r}(t)\), then present the MMSE estimation given \(\mathbb{O}_i(t)\).
In this section, we derive expressions for the MMSE estimators~\cite{stochetic} of $\tau^S_i(t)$ and $\tau^D_i(t+\theta_i)$ defined as $\hat{\tau}^S_i(t):=\mathbb{E}[\tau^{{S}}_i(t) | \mathbb{O}(t)]$ and $\hat{\tau}^D_i(t+\theta_i):=\mathbb{E}[\tau^D_i(t+\theta_i) | \mathbb{O}(t)]$, respectively. It follows that the MMSE estimators of $h_i(t+\theta_i)$ and $z_i(t)$ are given by
\begin{equation}\label{MMSEh}
    \hat{h}_i(t+\theta_i) = \mathbb{E}[h_i(t+\theta_i) | \mathbb{O}(t)] = t+\theta_i-\hat{\tau}^D_i(t+\theta_i)
\end{equation}
\begin{equation}\label{MMSEz}
    \hat{z}_i(t) = \mathbb{E}[z_i(t) | \mathbb{O}(t)] = t-\hat{\tau}^S_i(t)
\end{equation}
To reduce the complexity in computing $\hat{\tau}^S_i(t)$ and $\hat{\tau}^D_i(t+\theta_i)$, we develop Algorithm \ref{MWw/Estimationnofeedback} that iteratively computes the timestamp estimators over time. %$\hat{\tau}^S_i(t)$  and $\hat{\tau}^D_i(t+\theta_i)$ over time.

The estimator $\hat{\tau}^S_i(t)$ is based on $g_i(\varphi', t | \mathbb{O}(t))$. %, r
Recall that $g_i(\varphi', t | \mathbb{O}(t))$ is the PMF of ${\tau}^{{S}}_i(t)$ given $\mathbb{O}(t)$. Therefore, the MMSE estimator of $\tau^{{S}}_i(t)$  is given by
\begin{equation}\label{hattau^Snofeedback}    
    \hat{\tau}^{{S}}_i(t) = \sum_{\varphi'=\bar{\delta}_i[D_i(t)-1]+1}^t \varphi' g_i(\varphi',t| \mathbb{O}(t)) 
\end{equation}

The estimator \(\hat{\tau}^D_i(t+\theta_i)\) is based on \( g_i(\varphi', \varphi | \mathbb{O}(t)) \). Using \( g_i(\varphi', \varphi | \mathbb{O}(t)) \), the BS estimates the generation time of each received packet $d\in\{1,\ldots,D_i(t)\}$  from source \(i\). %with the observation \(\mathbb{O}(t)\) by \( g_i(\varphi', \varphi | \mathbb{O}(t)) \). 
However, since the channel between the BS and the destinations is unreliable, the latest packet received by the BS, i.e., the packet with index $D_i(t)$, may not be the latest packet received by destination $i$. %due to the unreliable channel between the BS and the destination, the index of the last delivered packet is not deterministic. 
Therefore, \( \hat{\tau}^D_i(t+\theta_i) \) should account for previous packets received by the BS. %can be jointly expressed by all transmitted packets prior to slot $t$.
% Consequently, we denote the MMSE estimation using the values of each estimation component and belief. 
% Then, by analyzing the observation, we provide a solution for belief.
%
Let \(\hat{\tau}^{{D}}_i[d]\) represent the MMSE estimator of the generation time of packet $d\in\{1,\ldots,D_i(t)\}$ received by the BS from source \(i\). % given the observation \(\mathbb{O}(t)\). 
Recall that \(g_i(\varphi', {\delta[d]} | \mathbb{O}(t))\) is the PMF of \(\hat{\tau}^{{D}}_i[d]\). Thus, \(\hat{\tau}^{{D}}_i[d]\) is given by
\begin{equation}\label{hattau^Ddnofeedback}
    \hat{\tau}^{{D}}_i[d] = \sum_{\varphi'=\bar{\delta}_i[d-1]+1}^{\delta[d]} \varphi' g_i(\varphi', {\delta[d]} | \mathbb{O}(t))
\end{equation}
Denote by \({b}_i[d]\) the probability that packet $d\in\{1,\ldots,D_i(t)\}$ is the latest packet delivered to destination $i$. The belief \({b}_i[d]\) is given by
\vspace{-0.5\baselineskip}
\begin{equation}
    {b}_i[d] = (1-(1-p^D_i)^{|{\mathcal{T}}_i[d]|} )\prod_{d'=d+1}^{D_i(t)} (1-p^D_i)^{|{\mathcal{T}}_i[d']|}
\end{equation}
\vspace{-0.5\baselineskip}
% assigned to \(\hat{\tau}^{{D}}_i[d]\), such that \(\sum_{d=1}^{D_i(t)}{b}_i[d]\leq1\). 
% The term \(\bar{b}_i := \mathbb{P}(\tau^D_i(t+\theta_i)=\tau_i^D(t-\omega_i))=1-\sum_{d=1}^{D_i(t)}{b}_i[d]\) represents the belief of \(\tau_i^D(t-\omega_i)\). 
Therefore, the MMSE estimator of \(\tau^D_i(t+\theta_i)\) is given by
% \begin{equation}\label{hattau^Drenew}
% \begin{aligned}
%     \hat{\tau}^D_i(t&+\theta_i) = \tau_i^D(t-\omega_i) \\&+\sum_{d=D_i(t-\theta_i-\omega_i)}^{D_i(t)} \left(\hat{\tau}^{{D}}_i[d] - \tau_i^D(t-\omega_i)\right) {b}_i[d],
% \end{aligned}    
% \end{equation}
\begin{equation}\label{hattau^Drenewnofeedback}
\hat{\tau}^D_i(t+\theta_i) = \sum_{d=1}^{D_i(t)} \hat{\tau}^{{D}}_i[d] {b}_i[d],
\end{equation}
\vspace{-0.7\baselineskip}
From (\ref{eq:D_i(t)}), (\ref{Tevolves}), and (\ref{hattau^Drenewnofeedback}), we can see that $\hat{\tau}^{{D}}_i(t+\theta_i)$ can be iteratively computed as follows. If $u_i(t-1)c^S_i(t-1)=1$, i.e., the BS successfully receives a packet from source $i$ during slot $t$, then:
\begin{equation}\label{hattau^D iter nofeedback}
\begin{aligned}
\hat{\tau}^{{D}}_i(t+\theta_i)=&(1-p_i^D)\mathbb{E}\left[\tau^{{D}}_i(t+\theta_i-1)|\mathbb{O}(t-1)\right]+\\&+p^D_i\hat{\tau}^D_i[D_i(t)] \; .
\end{aligned}
\end{equation}
Otherwise, if the BS does not receive a packet from source $i$ during slot $t$, then:
\begin{equation}\label{hattau^D iter nofeedback1}
\hat{\tau}^{{D}}_i(t+\theta_i)=\mathbb{E}\left[\tau^{{D}}_i(t+\theta_i-1)|\mathbb{O}(t-1)\right] \; .
\end{equation}
%expressed by $\bar{\tau}^{{D}}_i$, $\hat{\tau}^D_i[D_i(t)]$, and $c^S_i(t)$ as: 
% \begin{equation}\label{hattau^D iter nofeedback}
% \begin{aligned}
%  \hat{\tau}&^{{D}}_i(t+\theta_i) = \\
% &\begin{cases}
%     \bar{\tau}^{{D}}_i ,&\text{if $u_i(t-1)c^S_i(t-1)=0$}\\
%     \bar{\tau}^{{D}}_i(1-p_i^D)+\hat{\tau}^D_i[D_i(t)]p^D_i,&\text{if $u_i(t-1)c^S_i(t-1)=1$}
% \end{cases}  
% \end{aligned}
% \end{equation}
\begin{remark}
The iterative computation of $\hat{\tau}^{{S}}_i(t)$ and $\hat{\tau}^{{D}}_i(t+\theta_i)$ is described in Algorithm~\ref{MWw/Estimationnofeedback}. Algorithm~\ref{MWw/Estimationnofeedback} has a computational complexity of \(\mathcal{O}(N^2)\). For specific distributions of the packet generation process, the complexity can be further reduced. For example, for Bernoulli packet generation processes, \(\lambda_i(t)=\lambda_i, \forall t\), and for periodic packet generation processes with period \(\varGamma\), \(\lambda_i(t) = \mathbb{I}_{\{t\bmod\varGamma\}}\). This allows us to find closed-form expressions for $\hat{\tau}^{{S}}_i(t)$ and \(\hat{\tau}^D_i[d]\), thereby reducing the computational complexity from \(\mathcal{O}(N^2)\) to \(\mathcal{O}(1)\) for both distributions.
\end{remark}
%can also be further simplified by providing a closed-form expectation. For Bernoulli packet generation processes with a truncated geometric distribution, and for periodic packet generation processes, \(\tau^D_i[d] = \max_{t \leq \delta[d]} \{t\bmod\varGamma\} = 0\). Ultimately, the computational complexity can be reduced to \(\mathcal{O}(1)\) for both examples.
% Thus, at each given time slot, given the network parameters \(p^S_i, p^D_i, f_i\), system observation $\mathbb{O}(t)$, historical estimation $\bar{\tau}^{{D}}_i$ and PMF $g_i\left(\varphi', \varphi | \mathbb{O}(t-1)\right), \forall \varphi \leq t-1, \varphi' \leq \varphi$, the MMSE estimator is provided by Algorithm \ref{MWw/Estimationnofeedback}. 
\begin{algorithm}
\caption{MMSE Estimator of ${\tau}^{{S}}_i(t)$ and ${\tau}^{{D}}_i(t+\theta_i)$}
\label{MWw/Estimationnofeedback}
\begin{algorithmic}[1]
\STATE {\bf Input:} Observation: $\mathbb{O}(t)$; Historical PMF: $g_i(\varphi', \varphi | \mathbb{O}(t-1)), \forall \varphi \leq t-1, \varphi' \leq \varphi$; Historical Estimation: $ \mathbb{E}\left[\tau^{{D}}_i(t+\theta_i-1)|\mathbb{O}(t-1)\right]$
\STATE Update $\lambda_i(\varphi)$ and $g_i(\varphi',\varphi,t)$ using Algorithm \ref{alg:arrivalalgorithmnofeedback}
\STATE Compute $\hat{\tau}^{{S}}_i(t)$ using  (\ref{hattau^Snofeedback})
%\STATE \textbf{Step 3:} Compute $\hat{\tau}^{{D}}_i(t+\theta_i)$
\STATE Update $D_i(t)$ using (\ref{eq:D_i(t)})
\STATE Compute $\hat{\tau}^D_i[D_i(t)]$ using (\ref{hattau^Ddnofeedback})
\STATE Update $\hat{\tau}^{{D}}_i(t+\theta_i)$ using (\ref{hattau^D iter nofeedback}) and (\ref{hattau^D iter nofeedback1})
% \IF{$c^S_i(t)=1$}
% \STATE $\hat{\tau}^{{D}}_i(t+\theta_i)\leftarrow\bar{\tau}^{{D}}_i(1-p_i^D)+\hat{\tau}^D_i[D_i(t)]p^D_i$
% \ELSE 
% \STATE $\hat{\tau}^{{D}}_i(t+\theta_i)\leftarrow\bar{\tau}^{{D}}_i$
% \ENDIF
% \STATE Update $\tilde{\tau}^{{D}}_i$, $\tilde{\tau}^{{D}}_i[D_i(t)]$ using  (\ref{tildetau^D}), (\ref{tildetau^DD_i(t)})
\STATE {\bf Output:} $\hat{\tau}^{{S}}_i(t)$, $\hat{\tau}^{{D}}_i(t+\theta_i)$, $g_i(\varphi', \varphi | \mathbb{O}(t)),\forall\varphi \leq t, \varphi' \leq \varphi$
\end{algorithmic}
\end{algorithm}

Notice that Algorithm~\ref{MWw/Estimationnofeedback} does not take into account the delayed feedback received from the destinations. %In the case of finite feedback delay, 
Upon receiving feedback from a destination, the BS employs Algorithm~\ref{alg:arrivalalgorithmnofeedback} to recompute \(\lambda_i(\varphi)\) and \(g_i(\varphi', \varphi|\mathbb{O}(t))\).
Furthermore, 
%For estimation with feedback, when new feedback is received, i.e., when \(\tau^D_i(t-\theta_i-\omega_i) > \tau^D_i(t-\theta_i-\omega_i-1)\), the estimator requires the recalculation of \(\lambda_i(\varphi)\) and  \(g_i(\varphi', \varphi|\mathbb{O}(t))\) at each slot \(\varphi \in [t-\theta_i-\omega_i, t]\). Additionally, 
only packets transmitted after \(\tau^D_i(t-\theta_i-\omega_i)\) will impact the estimation. Therefore, the BS replaces \(\hat{\tau}^D_i[d]\) for all \(d \leq D_i(t-\theta_i-\omega_i)\) with \(\tau^D_i(t-\theta_i-\omega_i)\), recalculates the $\hat{\tau}^D_i[d]$ for all \(d > D_i(t-\theta_i-\omega_i)\) and the belief of the packet with index \(D_i(t-\theta_i-\omega_i) + 1\). %Consequently, the computational complexity increases to \(\mathcal{O}(\omega_i)\) for large \(\omega_i\).

\section{Max-Weight Policy with Estimation}\label{sec:MW}
In this section, we develop an AoI-aware Max-Weight (MW) Policy based on Lyapunov Optimization~\cite{neely2022stochastic} that leverages both the MMSE estimators $\hat{\tau}^S_i(t)$ and $\hat{\tau}^D_i(t+\theta_i)$. %at each time slot, ultimately resulting in significant performance improvements for the EWSAoI. 
Then, we derive performance guarantees for this MW policy with estimation. 

%\subsection{Max-Weight Policy}
The Max-Weight Policy is designed to reduce the expected drift of the Lyapunov Function at every slot \(t\). We denote the linear Lyapunov Function as follows
\begin{equation}
    L(t) = \frac{1}{N}\sum_{i=1}^{N}\beta_i h_i(t)
\end{equation}
where $\beta_i>0$ is a parameter that can tuned to different network configurations. Recall that scheduling decisions in slot $t$ affect the AoI in slot $t+\theta_i+1$. Hence, we define the one-slot Lyapunov Drift as
%At each time slot \(t\), the BS schedules transmissions selecting a set of \(u_i(t)\) for all \(i\). These decisions will affect the Age of Information of each source \(i\) in the time slot \(t+\theta_i+1\), indicated by \(h_i(t+\theta_i+1)\). Analyzing the 1-slot drift, we obtain the Lyapunov Drift as:
\begin{equation}\label{1-slot drift}
\begin{aligned}
    \Delta&(\mathbb{O}(t)) = \mathbb{E} \left[ L(t+\theta_i+1) - L(t+\theta_i) | \mathbb{O}(t) \right] \\
    &= \frac{1}{N}\sum_{i=1}^N\beta_i \mathbb{E} \left[ h_i(t+\theta_i+1) - h_i(t+\theta_i) | \mathbb{O}(t) \right]
\end{aligned}
\end{equation}
where $\mathbb{O}(t)$ is the BS observation at the beginning of slot $t$.
%
% From \ref{hevolve}, the value of $h_i(t+\theta_i+1) - h_i(t+\theta_i)$ rely on \(z_i(t)\) and ${h}_i(t+\theta_i)$.
% In our model, the \(z_i(t)\) is not available due to channel interference. Moreover, the ${h}_i(t+\theta_i)$ is future information which is also not available. Therefore, in the network model we consider, the BS needs to perform estimation and prediction of current and future AoI through observation \(\mathbb{O}(t)\) to realize dynamic scheduling policies.
% Denote \(\hat{h}_i(t+\theta_i)\) as the MMSE estimation of \(h_i(t+\theta_i)\) given the observation \(\mathbb{O}(t)\). Denote \(\hat{z}_i(t)\) as the MMSE estimation of \(z_i(t)\) given the observation \(\mathbb{O}(t)\). The estimators are given by:
% \begin{equation}\label{MMSEh}
%     \hat{h}_i(t+\theta_i) = \mathbb{E}[h_i(t+\theta_i) | \mathbb{O}(t)] 
% \end{equation}
% \begin{equation}\label{MMSEz}
%     \hat{z}_i(t) = \mathbb{E}[z_i(t) | \mathbb{O}(t)]
% \end{equation}
%
Substituting the evolution of \(h_i(t)\) from (\ref{hevolve}) into (\ref{1-slot drift}) and then manipulating the resulting expression, %since \(u_i(t)\) is independent of \(\hat{h}_i(t+\theta_i)\) and \(\hat{z}_i(t)\) given the observation \(\mathbb{O}(t)\), 
we obtain (\ref{Lyapunov Drift}). Notice that, given the observation \(\mathbb{O}(t)\), scheduling decisions \(u_i(t)\) are independent of the MMSE estimators \(\hat{h}_i(t+\theta_i)\) and \(\hat{z}_i(t)\). 

\begin{figure*}[htbp]
\begin{equation}\label{Lyapunov Drift}
\begin{aligned}
    \Delta \left( \mathbb{O}(t)\right) &= \frac{1}{N} \sum_{i=1}^N  \beta_i -\frac{1}{N} \sum_{i=1}^N \beta_i p^S_ip^D_i \mathbb{E}\left[\left(h_i(t+\theta_i) - z_i(t)-\theta_i\right) u_i(t) | \mathbb{O}(t)\right] \\
    &=\frac{1}{N} \sum_{i=1}^N  \beta_i -\frac{1}{N} \sum_{i=1}^N  \beta_i p^S_ip^D_i (\hat{h}_i(t+\theta_i) - \hat{z}_i(t)-\theta_i) \mathbb{E}\left[u_i(t) | \mathbb{O}(t)\right]
\end{aligned}
\end{equation}
\hrulefill
\end{figure*}

To minimize the one-slot Lyapunov Drift in (\ref{Lyapunov Drift}), the MW policy selects, in each slot \(t\), the \(K\) sources with highest values of \(\beta_i p^S_ip^D_i (\hat{h}_i(t+\theta_i) - \hat{z}_i(t)-\theta_i)\), with ties broken arbitrarily. 
%Since the estimation is subject to inaccuracies, the MW policy is not work-conserving. The MW policy is work-conserving unless the estimation is accurate when there is a deterministic arrival of each source and either the channel reliabilities $p^D_i=1, \forall i$ or both the transmission delay and the feedback delay subject to $\theta_i=\omega_i=0,\forall i$. 
For the case of no feedback from the destinations, i.e., \(\omega_i\rightarrow\infty,\forall i\), the MMSE estimator $\hat{h}_i(t+\theta_i)$ estimates AoI without actually receiving any information about the AoI at the destination. 
To the best of our knowledge, this is the first AoI-aware scheduling policy that attempts to minimize AoI without requiring any AoI knowledge.
Next, we derive performance guarantees for this MW policy.

% In order to obtain the $\hat{\tau}^S_i(t)$ and $\hat{\tau}^D_i(t+\theta_i)$, by definition, we need to obtain their PMFs separately. Next, we derive the PMF of the generation time of each packet given the observation $\mathbb{O}(t)$ to estimate $\tau^S_i(t)$ and $\tau^D_i(t+\theta_i)$.

% Hence, the Max-Weight Policy is not work-conserving, unless there is a deterministic arrival of every source. 

% Existing dynamic scheduling policies mainly rely on instantaneous AoI knowledge, i.e., \(\{h_i(t)\}_{i=1}^N\) and \(\{z_i(t)\}_{i=1}^N\). However, this assumption is not realistic due to transmission delays and channel interference. 

% Moreover, even if the instantaneous AoI is available, it may not provide accurate information in the presence of transmission delays because the AoI at the destination does not account for possible ongoing transmissions within \(\theta_i\). Therefore, in the network model we consider, the BS needs to perform estimation and prediction of current and future AoI through imperfect information observed in the network to realize dynamic scheduling policies.

%\subsection{Performance Analysis}

%Next, we derive performance guarantees for the MW policy. 
Theorem~\ref{theo:upperboundMW} provides an upper bound on the EWSAoI achieved by the MW policy. This upper bound is tighter than the bound provided in~\cite{liu2024optimizing} for the special case of perfect knowledge of AoI, i.e.,  \(\theta_i=0,p^D_i=1,\forall i\). Theorem~\ref{theo:upperboundMW} holds not only for the case of perfect knowledge of AoI, but also for imperfect knowledge, i.e., \(\theta_i>0, p^D_i<1, \omega_i<\infty\), and no knowledge, i.e., \(\omega_i\rightarrow\infty\). Theorem~\ref{theo:upperboundMW} and the simulation results discussed in Sec.~\ref{sec:simulation} suggest that MW outperforms the Optimal Randomized Policy in terms of EWSAoI even for the case of no knowledge of AoI. 

%which improves the upper bound at the case of \(\theta_i=0,p^D_i=1,\forall i\) given in \cite{liu2024optimizing}. Note that Theorem \ref{theo:upperboundMW} holds true even if $\omega_i\rightarrow \infty$, which shows that MW still outperforms optimal randomized policy in the case of no feedback.

\begin{theorem} \label{theo:upperboundMW}
The EWSAoI of the Max-Weight Policy is upper bounded by the EWSAoI of the Optimal Randomized Policy, which is formally expressed as
\begin{equation}
    L_B \leq \mathbb{E}\left[ J^{MW} \right] \leq \mathbb{E}\left[ J^{R} \right] \leq \rho L_B 
\end{equation}
where $\rho$ represents the optimality ratio given by Theorem \ref{theo:stationary_performance}.
\end{theorem}

\begin{proof}
The MW policy is designed to minimize the Lyapunov Drift (\ref{Lyapunov Drift}). By comparing the drift of MW policy with the drift of the Optimal Randomized Policy, we obtain
\vspace{-0.5\baselineskip}%ensuring the policy's scheduling efficiency by considering the observation $\mathbb{O}(t)$. Then we formulate the inequality
\begin{equation}
\begin{aligned}
     &\Delta \left( \mathbb{O}(t) \right) = \frac{1}{N}\sum_{i=1}^{N}\beta_{i}  \\&-\frac{1}{N}\sum_{i=1}^{N}\beta_{i}p^S_ip^D_i\left(\hat{h}_{i}(t+\theta_i) - \hat{z}_{i}(t)-\theta_i\right)\mathbb{E}\left[\mu_{i} |\mathbb{O}(t)\right]  \\&
     \leq  \frac{1}{N}\sum_{i=1}^{N}\beta_{i}  -\frac{1}{N}\sum_{i=1}^{N}\beta_{i}p^S_ip^D_i\left(\hat{h}_{i}(t+\theta_i) - \hat{z}_{i}(t)-\theta_i\right)\mu_{i}^R
\end{aligned}
\end{equation}

Taking the expectation with respect to \( \mathbb{O}(t) \) and computing the average over the time-horizon yields:
    \begin{equation}
    \begin{aligned}
     &\frac{\mathbb{E}[L(T+1)]}{T}-\frac{\mathbb{E}[L(1)]}{T} \leq \frac{1}{N} \sum_{i=1}^{N} \beta_i\\ & - \frac{1}{TN} \sum_{t=1}^{T} \sum_{i=1}^{N} \beta_i p^S_ip^D_i \left(\hat{h}_{i}(t+\theta_i) - \hat{z}_{i}(t)-\theta_i \right) \mu_i^R
    \end{aligned}
    \end{equation}

Substituting \( \beta_i = \alpha_i /(p^S_ip^D_i\mu_i^R) \), taking the limit as \( T \to \infty \), and then replacing \(\mathbb{E}[ \hat{h}_{i}(t+\theta_i) ]\) with \(\mathbb{E}[ \hat{h}_{i}(t) ]\) yields:
    \begin{equation}\label{between1}
    \begin{aligned}
    \lim_{T \to \infty}\frac{1}{TN} &\sum_{t=1}^{T} \sum_{i=1}^{N}  \alpha_i \mathbb{E}\left[ \hat{h}_{i}(t) \right]   \leq \frac{1}{N} \sum_{i=1}^{N} \frac{\alpha_i }{p^S_ip^D_i\mu_i^R}\\ & - \lim_{T \to \infty}\frac{1}{TN} \sum_{t=1}^{T} \sum_{i=1}^{N} \alpha_i \left(\mathbb{E}\left[ \hat{z}_{i}(t) \right]  +\theta_i\right)
    \end{aligned}
    \end{equation}

Applying the law of iterated expectations, we establish
\begin{equation}\label{E[h]}
\mathbb{E} [\hat{h}_i(t)] = \mathbb{E}\left[ \mathbb{E}\left[ h_{i}(t) | \mathbb{O}(t) \right] \right] = \mathbb{E} [h_i(t)]
\end{equation}
\begin{equation}\label{E[z]}
    \mathbb{E}\left[ \hat{z}_{i}(t) \right] = \mathbb{E}\left[ \mathbb{E}\left[ z_{i}(t) | \mathbb{O}(t) \right] \right] =\mathbb{E}\left[ z_{i}(t) \right]
\end{equation}
% , since we define the MMSE estimation in (\ref{MMSEh}),(\ref{MMSEz}) and use Law of Total Expectation for the substitution between (\ref{between1}) and (\ref{between2})
which leads to
    \begin{equation}\label{between2}
    \begin{aligned}
     \mathbb{E}\left[ {J}^{MW} \right]& \leq \frac{1}{N} \sum_{i=1}^{N} \frac{\alpha_i}{p^S_ip^D_i\mu_i^R} \\ & - \lim_{T \to \infty}\frac{1}{TN} \sum_{t=1}^{T} \sum_{i=1}^{N} \alpha_i \left(\mathbb{E}\left[ z_{i}(t) \right]  +\theta_i\right)
    \end{aligned}
    \end{equation}

%Notice that \eqref{between1}-\eqref{between2} supplement the missing steps between (48) and (49) of \cite{ji2024age}.

Substituting
    \begin{equation}
        \lim_{T \to \infty}\frac{1}{TN} \sum_{t=1}^{T} \sum_{i=1}^{N} \alpha_i \mathbb{E}\left[ z_{i}(t) \right] = \frac{1}{N} \sum_{i=1}^{N} \alpha_i \left(\frac{\mathbb{E}\left[X_i^{2}\right]\lambda_i}{2}-1\right)
    \end{equation}
into~\eqref{between2} gives
%Culminating in the final expression:
    \begin{equation}
    \begin{aligned}
        \mathbb{E}&\left[ {J}^{MW} \right]  \leq  \frac{1}{N} \sum_{i=1}^{N} \frac{\alpha_i}{ p^S_ip^D_i\mu_i^R}  \\&\quad+ \frac{1}{N} \sum_{i=1}^{N} \alpha_i \left(\frac{\mathbb{E}\left[X_i^{2}\right]\lambda_i}{2}+\theta_i-1\right) =\mathbb{E}\left[ J^{R} \right]
    \end{aligned}
    \end{equation}
\end{proof}
This result generalizes~\cite{kadota2019minimizing,ji2024age} to the more challenging network model with: (i) arbitrary packet generation processes at the sources; (ii) delayed and unreliable knowledge of the AoI at the destinations and timestamps at the sources; and (iii) transmission scheduling policies that can select multiple sources in each slot~$t$. 

\begin{figure*}[ht]
    \centering
    \begin{subfigure}[t]{0.32\textwidth}
        \centering
        \includegraphics[width=\textwidth]{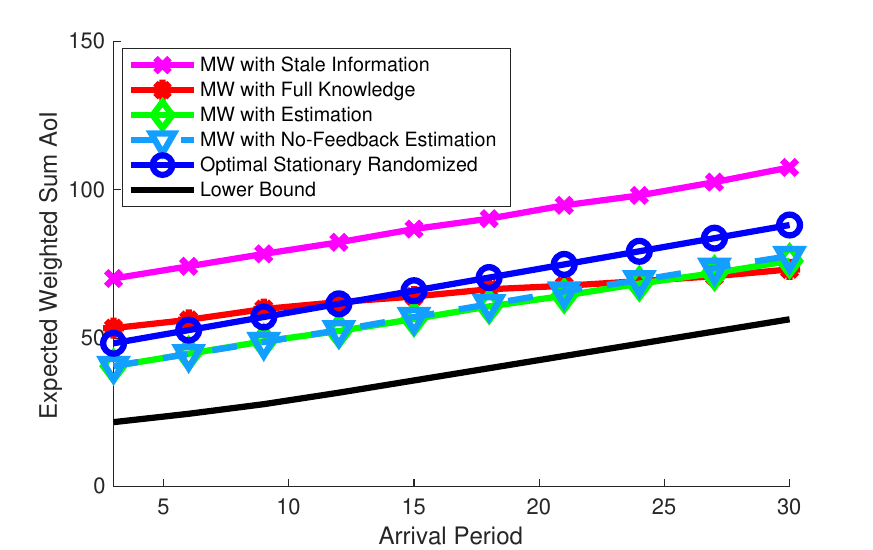}
        \caption{Uniform inter-generation periods}
        \label{fig:performance varying interarrival time uni}
    \end{subfigure}
    \hfill
    \begin{subfigure}[t]{0.32\textwidth}
        \centering
        \includegraphics[width=\textwidth]{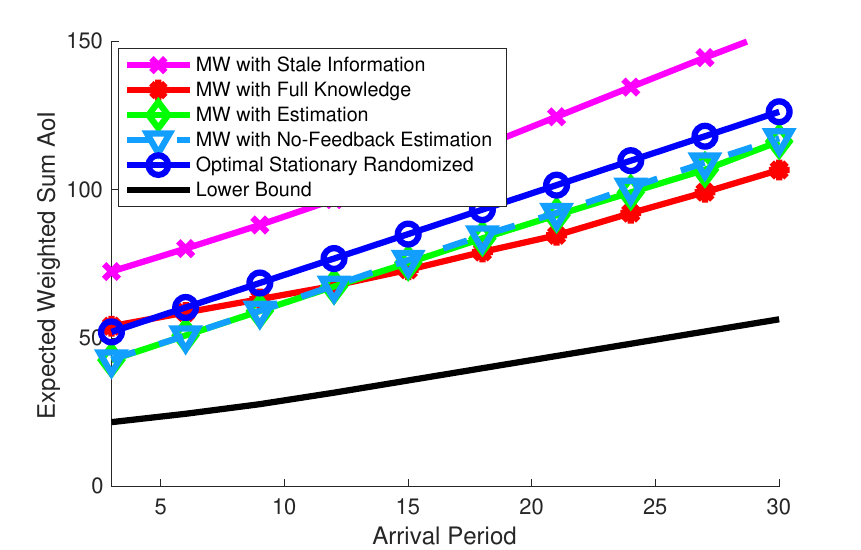}
        \caption{Bernoulli packet generation processes}
        \label{fig:performance varying interarrival time ber}
    \end{subfigure}
        \hfill
    \begin{subfigure}[t]{0.32\textwidth}
        \centering
        \includegraphics[width=\textwidth]{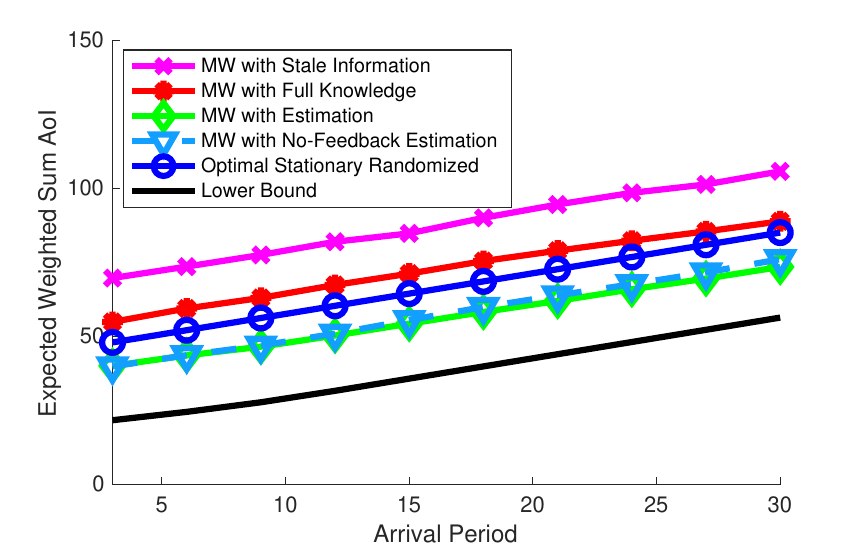}
        \caption{Periodic packet generation processes}
        \label{fig:performance varying interarrival time Per}
    \end{subfigure}
    \caption{Simulation results of networks with varying expected inter-generation periods $\mathbb{E}[X_i]$ with $N = 8,K = 2,$, $\boldsymbol{\alpha}=[4, 3, 2, 1, 5, 4, 1, 2]$, $ p^S_i = i/N, \forall i$, $ p^D_i = 0.8, \forall i$ and $\theta_i=5, \forall i$. The inter-generation period follows uniform distributions in (a), geometric distributions in (b), and is constant in (c).}
   
\end{figure*}

% \begin{figure*}[ht]
%     \centering
%     \begin{subfigure}[t]{0.48\textwidth}
%         \centering
%         \includegraphics[width=\textwidth]{Figure5_pS.pdf}
%         \caption{EWSAoI varying channel reliability $p^S_ip^D_i^S$}
%         \label{fig:performance varying channel reliability pS}
%     \end{subfigure}
%     \hfill
%     \begin{subfigure}[t]{0.48\textwidth}
%         \centering
%         \includegraphics[width=\textwidth]{Figure6_pD.pdf}
%         \caption{EWSAoI varying channel reliability $p_i^D$}
%         \label{fig:performance varying channel reliability pD}
%     \end{subfigure}
%     \caption{Simulation of network varying channel reliability with $N = 8,K = 2,$  $\boldsymbol{\alpha}=[4, 3, 2, 1, 5, 4, 1, 2]$ and $ X_i \sim U[2,4]$. The channel reliability $ p^D_i$ in fig.\ref{fig:performance varying channel reliability pS} is given by $ p^D_i = 0.8, \forall i$, the channel reliability $ p^S_i$ in fig.\ref{fig:performance varying channel reliability pD} is given by $ p^D_i = 0.8, \forall i$;
%     }
% \end{figure*}
% \begin{figure}
%     \centering
%     \includegraphics[width=0.48\textwidth]{Figure7.pdf}
%     \caption{Simulation of networks with  $N = 8,K = 2,$  $\boldsymbol{\alpha}=[4, 3, 2, 1, 5, 4, 1, 2]$ , $ X_i \sim U[2,4]$ and $ p^S_i = i/N, p^D_i = 0.8, \forall i$ varying transmission delay $\theta_i$. }
%     \label{fig:performance varying delay}
% \end{figure}

\begin{figure*}[ht]
    \centering
    \begin{minipage}[t]{0.65\textwidth}
        \centering
        \begin{subfigure}[t]{0.48\textwidth}
            \centering
            \includegraphics[width=\textwidth]{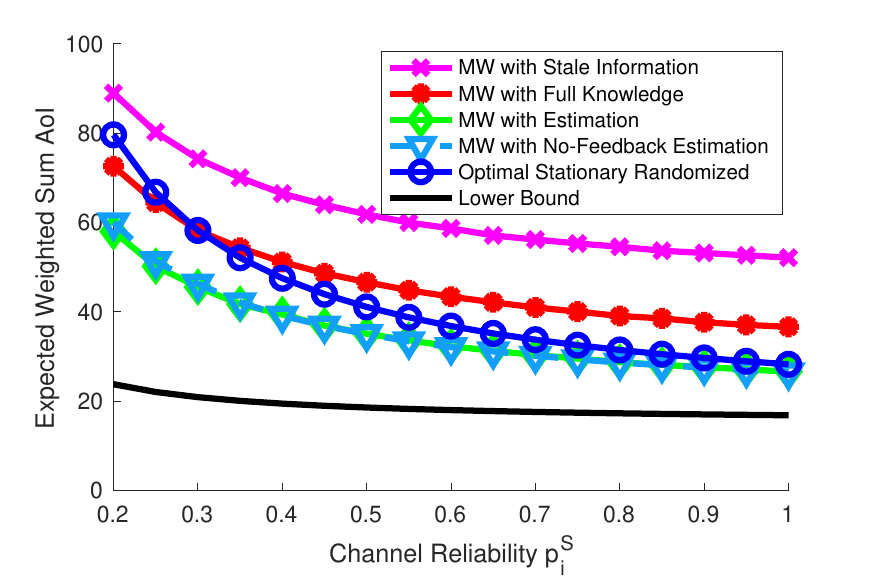}
            \caption{EWSAoI varying channel reliabilities $p_i^S$}
            \label{fig:performance varying channel reliability pS}
        \end{subfigure}
        \hfill
        \begin{subfigure}[t]{0.48\textwidth}
            \centering
            \includegraphics[width=\textwidth]{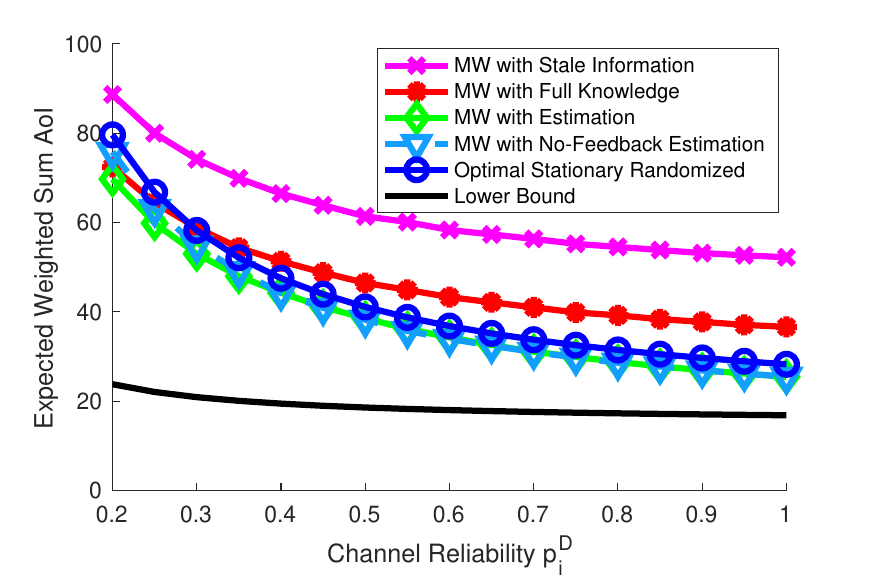}
            \caption{EWSAoI varying channel reliabilities $p_i^D$}
            \label{fig:performance varying channel reliability pD}
        \end{subfigure}
        \caption{Simulation results of networks with varying channel reliabilities with $N = 8, K = 2$, $\boldsymbol{\alpha}=[4, 3, 2, 1, 5, 4, 1, 2]$, $X_i \sim U[2,4]$ and $\theta_i=5, \forall i$. The channel reliabilities $p^D_i$ in (a) is given by $p^D_i = 0.8, \forall i$, and the channel reliabilities $p^S_i$ in (b) is given by $p^S_i = 0.8, \forall i$.}
        \label{fig:performance varying channel reliability}
    \end{minipage}
    \hfill
    \begin{minipage}[t]{0.32\textwidth}
        \centering
        \includegraphics[width=\textwidth]{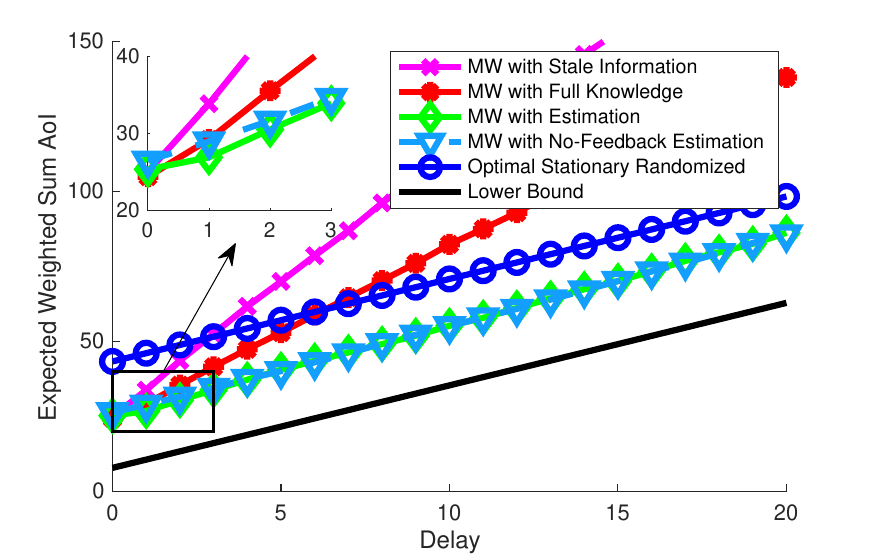}
        \caption{Simulation results of networks with varying transmission delay $\theta_i$ with $N = 8, K = 2$, $\boldsymbol{\alpha}=[4, 3, 2, 1, 5, 4, 1, 2]$, $X_i \sim U[2,4]$ , $p^S_i = i/N, \forall i$, $p^D_i = 0.8, \forall i$ and $\theta_i=5, \forall i$.}
        \label{fig:performance varying delay}
    \end{minipage}
\end{figure*}
\section{Simulation Results}\label{sec:simulation}

In this section, we evaluate the performance of various scheduling policies in terms of their EWSAoI. We compare the following policies:
(i) Optimal Randomized Policy; 
(ii) Max-Weight Policy with Estimation (MW-E), where $\omega_i=\theta_i, \forall i$; 
(iii) Max-Weight Policy with Estimation in the case of no Feedback (MW-EnF), where $\omega_i\rightarrow \infty, \forall i$; 
(iv) Max-Weight Policy with Full Knowledge (MW-F), where the BS knows both the AoI at the destinations $h_i(t)$ and the system times at the sources $z_i(t)$ in real-time. In each slot $t$, MW-F selects the K sources with highest values of $\beta_i p^S_ip^D_i \left(h_i(t) - z_i(t)-\theta_i\right)$; 
(v) Max-Weight Policy with Stale Information (MW-S), where instead of estimating AoI and system times, the BS uses delayed AoI information and $\mathbb{E}[z_i] = 1/\lambda_i-1$. In each slot $t$, MW-S selects the K sources with highest values of $\beta_i p^S_ip^D_i \left(h_i(t-\omega_i) - \mathbb{E}[z_i]-\theta_i\right)$. 
We compare the performance of these policies with the lower bound derived in Sec.~\ref{sec:lower_bound}. %for any packet generation processes.
The performance of the Optimal Randomized Policy is computed using the expressions in Section IV, while the performances of  MW-E, MW-EnF, MW-F, and MW-S policies are averaged over 10 simulation runs, each with a time horizon of $T = 1 \times 10^6$ slots.

In Figs.~\ref{fig:performance varying interarrival time uni}, \ref{fig:performance varying interarrival time ber}, and \ref{fig:performance varying interarrival time Per}, we simulate networks with \(N = 8\) sources with priorities \(\boldsymbol{\alpha}:=[\alpha_1,\alpha_2,\ldots,\alpha_N]=[4, 3, 2, 1, 5, 4, 1, 2]\). The channel reliabilities are \(p^S_i = i/N,\forall i\), and \(p^D_i = 0.8, \forall i\), the maximum number of transmissions in each slot set to $K = 2$, transmission delays set to $\theta_i=5, \forall i$. In Fig. \ref{fig:performance varying interarrival time uni}, the inter-generation period \(X_i\) follows a uniform distribution with \(U[2x, 4x]\), where \(x \in \{1, 2, \ldots, 10\}\). In Fig. \ref{fig:performance varying interarrival time ber}, the inter-generation period \(X_i\) follows a geometric distribution with parameter \(\lambda_i\in \{1/3,1/6,\ldots,1/30\},\forall i\). %, i.e., Bernoulli packet generation processes. 
In Fig. \ref{fig:performance varying interarrival time Per}, the inter-generation period is constant with \(X_i\in \{3,6,\ldots,30\},\forall i\). % for all \(i\) is a constant, i.e., periodic packet generation process.

In Figs. \ref{fig:performance varying channel reliability pS} and \ref{fig:performance varying channel reliability pD}, we simulate networks with $N = 8$ sources, priorities $\boldsymbol{\alpha}=[4, 3, 2, 1, 5, 4, 1, 2]$, maximum number of transmissions in each slot set to $K = 2$, transmission delays set to $\theta_i=5, \forall i$, and inter-generation periods $X_i$ that follow $U[2, 4]$. The channel reliabilities vary as $p^S_i \in \{0.2, 0.25, ..., 1\}, \forall i$, with $p^D_i = 0.8, \forall i$, in Fig. \ref{fig:performance varying channel reliability pS}, and vary as $p^D_i \in \{0.2, 0.25, ..., 1\}, \forall i$, with $p^S_i = 0.8, \forall i$, in Fig. \ref{fig:performance varying channel reliability pD}. 

In Fig. \ref{fig:performance varying delay}, we simulate networks with $N = 8$ sources, priorities $\boldsymbol{\alpha}=[4, 3, 2, 1, 5, 4, 1, 2]$, and inter-generation periods $X_i$ that follow $U[2, 4]$. The channel reliabilities are $p^S_i =i/N, p^D_i = 0.8, \forall i$. The maximum number of transmissions in each slot is set to $K = 2$ and the transmission delay varies as $\theta_i \in \{1, 2, ..., 20\}, \forall i$. 
%The performance of the randomized policies is computed using the expressions in Section IV, while the performance of the MW-E, MW-EnF, MW-F and MW-S policies are averaged over 10 simulation runs, each with a time horizon of $T = 1 \times 10^6$ slots.

% (i) Optimal Randomized Policy; 
% (ii) Max-Weight Policy with Estimation (MW-E), where $\omega_i=\theta_i, \forall i$; 
% (iii) Max-Weight Policy with Estimation in the case of no Feedback (MW-EnF), where $\omega_i\rightarrow \infty, \forall i$; 
% (iv) Max-Weight Policy with Full Knowledge (MW-F), where the BS knows both the AoI at the destinations $h_i(t)$ and the system times at the sources $z_i(t)$ in real-time. In each slot $t$, MW-F selects the K sources with highest values of $\beta_i p^S_ip^D_i \left(h_i(t) - z_i(t)-\theta_i\right)$; 
% (v) Max-Weight Policy with Stale Information (MW-S), where instead of estimating AoI and system times, the BS uses delayed AoI information and $\mathbb{E}[z_i] = 1/\lambda_i-1$.

\begin{remark}
Numerical results show that Max-Weight Policies leveraging MMSE estimation (namely, MW-E and MW-EnF) outperform the Optimal Randomized Policy in every simulated scenario. 
%Moreover, they show that  MW-E and MW-EnF  approach the performance of MW-F when $\theta_i=0, \forall i$, regardless of channel reliabilities and packet generation processes. 
Interestingly, Max-Weight Policy with Full Knowledge of $h_i(t)$ falls behind MW-E, MW-EnF, and the Optimal Randomized Policy in many scenarios. 
This is because the delayed and unreliable channel between the BS and the destinations creates uncertainty regarding the reception of packets at the destinations. 
%Notice that the BS is at a vantage point with respect to the destinations. 
% Specifically, the BS knows in advance that packets were forwarded to the destinations. 
%Specifically, the BS probabilistically determines in advance whether the packets were forwarded to the destinations. 
Both MW-E and MW-EnF leverage their knowledge of channel reliabilities $p_i^D$ and the available information in the BS observation $\mathbb{O}(t)$ to estimate the future AoI $\hat{h}_i(t+\theta_i)$ at the destinations. In contrast, MW-F leverages its (impractical) real-time knowledge of $h_i(t)$. Critically, MW-F does not attempt to estimate the future AoI. 
%prevents the destination from immediately observing the transmission outcomes of the packets transmitted during slots $[t-\theta_i,t]$ at beginning of slot~$t$. 
\end{remark}

%Figs. \ref{fig:performance varying channel reliability pS} and \ref{fig:performance varying channel reliability pD} show that for the same overall channel reliabilities $p_i^S \times p^D_i$, the performance of MW-F, MW-S, and the Optimal Randomized Policy remain consistent. However, MW-E and MW-EnF achieve lower AoI when $p^D_i$ is high and $p^S_i$ is low, but higher AoI when $p^S_i$ is high and $p^D_i$ is low, because the BS can observe $c_i^S$ in real-time for estimation, while there is a delay in observing $c_i^D$, making $p^D_i$ more impactful on AoI for MW-E. 
%Additionally, the gap between the EWSAoI of MW-E and the randomized policy remains relatively stable and virtually independent of delay $\theta_i$, as only the transmission outcome affects the estimation when the delay is large enough. 
\begin{remark}
Numerical results also show that the performance of MW-E and MW-EnF are similar in every simulated scenario. This suggests that the MMSE estimators that do not leverage the feedback from the destinations provide sufficient accuracy for optimizing AoI. 
%are almost as accurate as MMSE estimators that leverage such feedback. 
Since the AoI is a (relatively) recent metric that is not currently measured in existing base stations, transmission scheduling algorithms that can optimize AoI without measuring it can be of particular interest for system implementation.
\end{remark}
%Figure \ref{fig:performance varying delay} shows that in the absence of feedback, MW-EnF can provide almost the same performance with MW-E, ensuring information freshness even when the BS schedules without any AoI information.
%This not only minimizes AoI but also reduces the resource overhead caused by sharing AoI knowledge in complex networks.

\section{Final Remarks}\label{sec:conclusion}

This paper studies the problem of optimizing scheduling decisions for minimizing AoI without relying on perfect knowledge of AoI. Specifically, we consider a network with sources that generate packets according to general renewal processes, a wireless BS that schedules multiple unreliable packet transmissions at every time slot $t$, and destinations with connections to the BS that are unreliable and delayed. 
%This paper investigates a network where a BS serves multiple sources to various destinations. The packets from each source arrive at the BS following an arbitrary renewal process. We address the problem of optimizing scheduling decisions concerning the EWSAoI of the network.
Our main contributions are as follows:
(i) We derive a lower bound on the achievable EWSAoI; 
(ii) We develop a Optimal Randomized Policy for arbitrary renewal packet generation processes; 
(iii) We obtain MMSE estimators of system time and AoI; 
(iv) We develop a Max-Weight Policy that leverages these MMSE estimators; 
(v) We evaluate the EWSAoI of the Optimal Randomized Policy and of the Max-Weight Policy both analytically and through simulations.
%Our results demonstrate that the performance of the Max-Weight Policy with MMSE estimation is close to the analytical lower bound and is virtually independent of transmission delays. Crucially, the MW-E can guarantee information freshness even when the BS schedules without direct AoI information. 
Interesting extensions include consideration of estimates of the packet generation processes at the sources, consideration of non-i.i.d. and correlated wireless channels, consideration of time-varying transmission/feedback delays, and experimental evaluation using software-defined radios.

\bibliographystyle{IEEEtran}
\bibliography{references}

\end{document}